%% file: ms.tex
\documentclass[preprint,numbers,nocopyrightspace]{sigplanconf}

\usepackage[utf8]{inputenc}
\usepackage[T1]{fontenc}
\usepackage{microtype}

\usepackage{amsmath,amssymb,amsfonts,amsthm}
\usepackage{stmaryrd,wasysym,pifont,marvosym}
\usepackage{thmtools}
\usepackage{mathtools}
\usepackage{graphicx}
\usepackage[usenames,dvipsnames]{color,xcolor}
\usepackage[inline]{enumitem}
\usepackage{url}
\usepackage{adjustbox}
\usepackage{booktabs}
\usepackage{bm}
\usepackage{xspace}
\usepackage[xspace]{ellipsis}
\usepackage{hyphenat}
\usepackage{setspace}
\usepackage{etoolbox}
\usepackage{scalerel}
\usepackage{twoopt}
\usepackage{mfirstuc}
\usepackage{multirow,bigdelim}
\usepackage[labelfont=bf, skip=5pt]{caption}
\usepackage{subcaption}
\usepackage{algorithm}
\usepackage[noend]{algpseudocode}
\usepackage[scaled=0.80]{beramono} 
\usepackage[scaled=0.81]{berasans} 
\usepackage[color=red!55,textsize=footnotesize,textwidth=1.3cm,shadow,prependcaption]{todonotes}
\usepackage{mdframed}
\usepackage{relsize}
\usepackage{alltt}
\usepackage{tikz}
\usetikzlibrary{arrows.spaced,positioning,calc,quotes}
\usepackage{pgfplots}
\usepackage{colortbl}
\pgfplotsset{compat=1.12}
\pgfplotsset{every tick label/.append style={font=\footnotesize}}
\usepackage{pgfplotstable}
\usepackage{listings}
\lstalias[]{csharp}[Sharp]{C}
\lstdefinelanguage{dsl}{
    morekeywords={language,feature,using,semantics,learners,int,string,@input,@start,values,let,in,@ref,@values,Tuple,@extern,@output,namespace,bool},
    otherkeywords={:=,|,;,=>,:,[,]},
    sensitive=true,
    morecomment=[l]{//},
    morestring=[b]',
}

\usepackage[colorlinks=true,allcolors=blue,breaklinks]{hyperref}
\usepackage[capitalise,nameinlink]{cleveref}
\hypersetup{final}

\makeatletter
    \crefname{section}{\S\@gobble}{\S\@gobble}
    \crefname{defn}{definition}{definitions}
    \pretocmd{\NAT@citexnum}{\@ifnum{\NAT@ctype>\z@}{\let\NAT@hyper@\relax}{}}{}{}
    \DeclareRobustCommand\em
        {\@nomath\em \ifdim \fontdimen\@ne\font >\z@
        \bfseries \else \itshape \fi}
    \newenvironment{xflalign*}
        {\setlength{\abovedisplayskip}{0pt}\setlength{\belowdisplayskip}{0pt}%
         \csname flalign*\endcsname}
        {\csname endflalign*\endcsname\ignorespacesafterend}
\makeatother

\AtBeginEnvironment{algorithmic}{\setstretch{1.16}}

\definecolor{linenum-gray}{HTML}{888888}

\newcommand{\NoNumber}{\def\alglinenumber##1{}}
\newcommand{\WithNumber}{\def\alglinenumber##1{\sf\scriptsize\color{linenum-gray}##1:\hspace{-11pt}}}
\newcommand{\Functionx}[2]{\NoNumber \Function{#1}{#2} \addtocounter{ALG@line}{-1} \WithNumber}
\let\oldStatex\Statex
\renewcommand{\Statex}{\oldStatex \hspace{\algorithmicindent}}

\DeclareMathOperator*{\argmax}{\arg\!\max}

\DeclareMathOperator*{\bigvsaunion}{\adjustbox{scale=1.6,valign=t}{$\pmb{\mathsf{U}}$}}

\newtheorem{theorem}{Theorem}

\newtheorem{defn}{Definition}
\theoremstyle{definition}
\newtheorem{example}{Example}
\newtheorem{problem}{Problem}

\begin{document}
\input{macros}

\conferenceinfo{POPL '17}{}
\copyrightyear{2017}

\title{Interactive Program Synthesis}
\authorinfo{Vu Le}{Microsoft Corporation}{levu@microsoft.com}
\authorinfo{Daniel Perelman}{Microsoft Corporation}{danpere@microsoft.com}
\authorinfo{Oleksandr Polozov}{University of Washington}{polozov@cs.washington.edu}
\authorinfo{Mohammad Raza}{Microsoft Corporation}{moraza@microsoft.com}
\authorinfo{Abhishek Udupa}{Microsoft Corporation}{abudup@microsoft.com}
\authorinfo{Sumit Gulwani}{Microsoft Corporation}{sumitg@microsoft.com}
\maketitle

\input{abstract}
\input{introduction}
\input{background}

\input{overview}
\input{incremental}
\input{step}

\input{feedback}

\input{evaluation}

\input{related}
\input{conclusion}

\bibliographystyle{abbrvnat}
\bibliography{PROSE}

\end{document}

%% file: macros.tex
\newcommand{\ignore}[1]{}
\newcommand{\discuss}[1]{{\color{red}#1}}

\newcommand{\projectName}{PROSE\xspace}
\newcommand{\topDownName}{\ensuremath{D^{4}}\xspace}
\newcommand{\fe}{FlashExtract\xspace}
\newcommand{\ff}{FlashFill\xspace}
\newcommand{\fs}{FlashSplit\xspace}
\newcommand{\hypothesizer}{hypothesizer\xspace}
\newcommand{\decomposable}{definitive\xspace}
\newcommand{\localRefine}{locally refining\xspace}
\newcommand{\globalRefine}{globally refining\xspace}

\newcommand{\fetests}{100\xspace}

\newcommand{\dsl}{\mathcal{L}}
\newcommand{\state}{\sigma}
\newcommand{\inputSymbol}[1]{\mathsf{input}(#1)}
\newcommand{\start}[1]{\mathsf{output}(#1)}
\newcommand{\rhs}[1]{\mathsf{RHS}(#1)}
\newcommand{\freeVariables}[1]{\mathsf{FV}(#1)}
\newcommand{\spec}{\varphi}
\newcommand{\constraint}{\psi}
\newcommand{\synalgorithm}{\mathcal{A}}
\newcommand{\rank}{h}
\newcommand{\join}{\Bowtie}
\newcommand{\volume}[1]{V(#1)}
\newcommand{\width}[1]{W(#1)}
\newcommand{\vsa}[1][N]{\widetilde{#1}}
\newcommand{\values}[1][v]{\pvec{#1}}
\newcommand{\states}[1][\state]{\pvec{#1}}
\newcommand{\vsajoin}[1][F]{#1_{\text{\tiny\join}}}
\newcommand{\wf}[1][]{\,\mathrel{\pmb{\Longrightarrow}_{\!\!#1}}\,}
\newcommand{\wfc}[1][]{\,\mathrel{\pmb{\Longleftrightarrow}_{\!\!#1}}\,}
\newcommand{\dep}[1][]{\mathrel{\pmb{\models}}}
\newcommand{\result}{\semantics{\cdot}}
\newcommand{\atsign}{\makeatletter @\makeatother}
\newcommand{\inv}{{\scriptscriptstyle-\!1}}
\newcommand{\tospec}{\rightsquigarrow}
\newcommand{\Bool}{\mathsf{Bool}}
\newcommand{\Reals}{\mathbb{R}}
\newcommand{\true}{\mathsf{TRUE}}
\newcommand{\false}{\mathsf{FALSE}}
\newcommand{\field}{e}
\newcommandtwoopt{\disambScore}[2][q][\vsa]{\mathsf{ds}(#1 , #2)}

\newcommand{\vsaunion}{\pmb{\mathsf{U}}}
\newcommand{\vsaunionop}{\mathbin{\vsaunion}}
\newcommand{\semantics}[1]{\llbracket #1 \rrbracket}
\newcommand{\is}{\;:=\;}
\newcommand{\palt}{\;|\;}
\newcommand{\bydef}{\mathrel{\overset{\mathsf{\scriptscriptstyle def}}{=}}}
\newcommand{\todoinline}[1]{\vspace*{\topsep}\todo[inline, caption={\textbf{\textsf{TODO}}}]{#1}}
\newcommand{\pvec}[1]{\vec{#1}\mkern2mu\vphantom{#1}}
\newcommand{\bigmid}{\mathrel{\big|}}
\newcommand{\feseqdsl}{\dsl_{\mathsf{ES}}}
\newcommand{\fedsl}{\dsl_{\mathsf{FE}}}
\newcommand{\fsdsl}{\dsl_{\mathsf{FS}}}
\newcommand{\ffdsl}{\dsl_{\mathsf{FF}}}
\newcommand{\yesmark}{\ding{51}}
\newcommand{\nomark}{\ding{55}}
\newcommand{\hlc}[2][box-blue]{{\sethlcolor{#1} \hl{#2}}}
\newcommand{\stringliteral}[1]{\ensuremath{\text{\begin{small}``\texttt{#1}''\end{small}}}}
\newcommand{\assuming}{\;\ifnum\currentgrouptype=16 \middle\fi\vert\;}

\newcommandtwoopt{\projection}[2][\vsa][\spec]{{#1\hspace{-0.5pt}\stretchrel*{\upharpoonright}{#1_.}}_{\raisebox{1pt}{\smaller $#2$}}}
\newcommandtwoopt{\clustering}[2][\vsa][\state]{#1 /_{#2}}

\robustify{\stringliteral}

\pgfplotstableset{
    /color cells/min/.initial=0,
    /color cells/max/.initial=1000,
    /color cells/textcolor/.initial=,
    %
    color cells/.code={%
        \pgfqkeys{/color cells}{#1}%
        \pgfkeysalso{%
            postproc cell content/.code={%
                \begingroup
                %
                \pgfkeysgetvalue{/pgfplots/table/@preprocessed cell content}\value
                \ifx\value\empty
                    \endgroup
                \else
                \pgfmathfloatparsenumber{\value}%
                \pgfmathfloattofixed{\pgfmathresult}%
                \let\value=\pgfmathresult
                %
                \pgfplotscolormapaccess
                    [\pgfkeysvalueof{/color cells/min}:\pgfkeysvalueof{/color cells/max}]
                    {\value}
                    {\pgfkeysvalueof{/pgfplots/colormap name}}%
                %
                \pgfkeysgetvalue{/pgfplots/table/@cell content}\typesetvalue
                \pgfkeysgetvalue{/color cells/textcolor}\textcolorvalue
                %
                \toks0=\expandafter{\typesetvalue}%
                \xdef\temp{%
                    \noexpand\pgfkeysalso{%
                        @cell content={%
                            \noexpand\cellcolor[rgb]{\pgfmathresult}%
                            \noexpand\definecolor{mapped color}{rgb}{\pgfmathresult}%
                            \ifx\textcolorvalue\empty
                            \else
                                \noexpand\color{\textcolorvalue}%
                            \fi
                            \the\toks0 %
                        }%
                    }%
                }%
                \endgroup
                \temp
                \fi
            }%
        }%
    }
}

%% file: abstract.tex
\begin{abstract}
    Program synthesis from incomplete specifications (e.g. input-output
    examples) has gained popularity and found real-world applications,
    primarily due to its ease-of-use.
    Since this technology is often used in an interactive
    setting, efficiency and correctness are often the key user expectations
    from a system based on such technologies.
    Ensuring efficiency is challenging, since the highly combinatorial nature of program synthesis
    algorithms does not fit in a 1--2 second response expectation of a user-facing system.
    Meeting correctness expectations is also difficult, given
    that the specifications provided are incomplete, and that the users of such
    systems are typically non-programmers.

    In this paper, we describe how \emph{interactivity} can be leveraged
    to develop efficient synthesis algorithms, as well as to decrease
    the cognitive burden that a user endures trying to ensure that the
    system produces the desired program.
    We build a formal model of user interaction along three dimensions:
    \emph{incremental algorithm}, \emph{step-based problem formulation}, and
    \emph{feedback-based intent refinement}.
    We then illustrate the effectiveness of each of these forms of
    interactivity with respect to synthesis performance and correctness on a set of real-world case studies.
\end{abstract}

%% file: introduction.tex
\section{Introduction}
\label{sec:introduction}
Program synthesis is the task of generating a program in an underlying
\emph{domain-specific language} (DSL) from an intent specification
provided by a user~\cite{gulwani2010dimensions}.
When the user in question is a non-programmer,
the specification method must be concise and easy to provide without
any programming expertise.
The \emph{programming by examples} (PBE) paradigm,\footnote{In this
  article, whenever applicable, we use PBE to more generally denote
  programming using under-specifications.} where the user intent is
specified by the means of input-output examples or constraints,
satisfies this requirement ideally.
Although examples are succinct and easy for users to provide, they form an
\emph{under-specification} on the behavior of the desired program,
adding inherent ambiguity into the problem definition.

Thanks to its ease of use, PBE has been effectively applied to many real-world scenarios in mass-market deployments.
Two prominent examples are \ff~\cite{flashfill} and \fe~\cite{flashextract}.
\ff is a technology for automating repetitive string transformations, released as a feature in Microsoft Excel 2013.
\fe is a technology for extracting hierarchical data from semi-structured text files, released for log analytics in Microsoft Operations
Management Suite and as the \textsf{ConvertFrom-String} cmdlet in Windows PowerShell.

The large-scale and continued adoption of PBE techniques is contingent
on ensuring \textbf{(a)} the performance of the synthesizer, and
\textbf{(b)} the correctness of the synthesized program.
Responsiveness is critical to making PBE technologies usable.
Users are willing to interact with the system in many rounds providing constraints iteratively, but any wait time exceeding 1--2 seconds per
round leads to a frustrating experience.
Deployed systems like \ff and \fe ensure performance by restricting
the expressiveness of the underlying DSL or by
bounding the execution time of the synthesizer.
Such restrictions limit the applicability and growth of these
technologies:
when the underlying DSL is enriched to meet the users' demands for
capturing a larger class of tasks, the performance of the synthesizer starts
degrading.

The correctness of the synthesized program is critical to building trust in PBE systems.
In the past, intent ambiguity in PBE has been primarily handled by imposing a sophisticated ranking on the DSL~\cite{cav:ranking}.
While ranking goes a long way in avoiding undesirable interpretations
of the user's intent, it is not a complete solution.
For example, \ff in Excel is designed to cater to users that care not about the program but about its behavior on the small number of input
rows in the spreadsheet.
Such users can simply eye-ball the outputs of the synthesized program and provide another example if they are incorrect.
However, this becomes much more cumbersome (or impossible) with a larger spreadsheet.%
\footnote{John Walkenbach, famous for his Excel textbooks, labeled \ff as a ``controversial'' feature.
He wrote: \textsf{``It's a great concept, but it can also lead to lots of bad data. I think many users will look at a few ``flash filled''
cells, and just assume that it worked. But my preliminary tests leads me to this conclusion: Be very
careful.''}~\cite{walkenbach:controversial}}

We have observed that inspecting the synthesized program directly also does not establish enough confidence in it even if the user knows
programming.
Two main reasons for this are (i) program readability,%
\footnote{Stephen Owen, a certified MVP (``Most Valued Professional'')
  in Microsoft technologies, said the following of a program
  synthesized by \fe:
\textsf{``If you can understand this, you’re a better person than I am.''}~\cite{mvp:complaint}} and (ii) the users' uncertainty in the
desired intent due to hypothetical unseen corner cases in the data.
This feedback was consistent with a recent user study~\cite{flashprog}, which illustrated that users find it less useful and approachable to
inspect the synthesized programs but would prefer more interactive models to converge to the desired program and to be confident of its
correctness.

Due to ambiguity of intent in PBE, the standard user interaction model in this setting is for the user to provide constraints
\emph{iteratively} until the user is satisfied with the synthesized program or its behavior on the known inputs.
However, most work in this area, including \ff and \fe, has not been formally modeled as an iterative process.
In this paper, we propose an interactive formulation of program synthesis that leverages the inherent iterative nature of
synthesis from under-specifications.

\subsection*{Interactivity}
We present \emph{interactivity} as the solution to addressing the performance and the correctness challenges associated with PBE.
Our inspiration to make program synthesis interactive comes from the standard world of programming.
In programming, interactivity manifests in at least three key dimensions:
\begin{description}[topsep=2pt, itemsep=1pt]
    \item[Incremental:] A programmer writes a function by iteratively refining it, for instance, as in test-driven development.
    \item[Step-based:] A programmer splits the task of writing a program into simpler steps of writing various functions for individual
        sub-tasks one by one.
    \item[Feedback-based:] A programmer may use various tools, such as program analysis or test coverage measurement, to obtain
        actionable feedback on code quality, correctness, and performance.
\end{description}
We propose integrating these dimensions intro the interactive process of program synthesis.

\paragraph{Incremental synthesis}
The standard PBE model requires the user to refine her intent in iterative rounds by providing additional constraints on the current
candidate program.
The standard approach has been to re-run the synthesizer afresh with the conjunction of the original constraints and the new
constraints.
In this paper, we describe an alternative technique, which makes the synthesizer \emph{incremental} and provides significant
performance benefits.

Most PBE techniques use a data structure called \emph{version space algebra} (VSA)~\cite{flashmeta} to succinctly represent and compute
the set of programs in the underlying DSL that are consistent with the user-provided constraints.
Our key idea is to observe that a VSA is simply \emph{an Abstract
  Syntax Tree (AST) based representation of a language}, and \emph{it can be translated into a sub-DSL of
the original DSL}.
The new round of synthesis is then carried out over this new sub-DSL using only the new constraints.

\paragraph{Step-based synthesis}
In many PBE domains, there exist well-defined sub-computations that can be exposed to the user intuitively as steps.
This helps with both synthesizer performance and ensuring correctness of the synthesized program.

For instance, the programs in the \ff DSL apply a conditional logic to compute different string transformations based upon the input string.
When the user provides a few input-output examples, the \ff synthesizer conjectures which of those examples will be addressed by the same branch of the top-level conditional.
This decision is based on a few examples and can be incorrect.
Furthermore, for a complicated task that requires many examples, the conditional learning algorithm in \ff does not
scale~\cite{wu2015iterative,kini15}.

The conditional logic is a sub-computation that can be naturally
exposed to the user as a clustering of input rows.
The user can drive this task by providing examples of rows that should be in the same cluster.
Our key idea to enable step-based synthesis is to \emph{allow associating user constraints with named subexpressions} in the DSL.

\paragraph{Feedback-based synthesis}
In the standard PBE model, the user is responsible for providing additional constraints in each iterative round of synthesis.
In each iteration, the user chooses what additional constraints to
provide, based purely on the behavior of the synthesized program, but
without any directed feedback from the synthesizer.
However, the synthesizer has knowledge about the specific ambiguities in the constraints provided by the user w.r.t. the underlying DSL.
Our key idea is to translate this knowledge into \emph{natural queries for the user}, whose responses \emph{form the next set of additional
constraints}.
This provides two key benefits:
\begin{enumerate}[label=(\alph*), topsep=3pt, itemsep=1pt]
    \item The feedback from the tool can help continue the progress towards the desired program by pointing out remaining ambiguities.
        Otherwise, the user may stop earlier than intended.
    \item The new set of constraints constructed from the user's response can accelerate the progress towards the desired program by
        resolving the ambiguity faster (relative to the constraints that the user would have provided otherwise).
        To this end, we introduce a novel automatic component in the conventional PBE model, called \emph{the \hypothesizer}.
        It proactively analyzes the current set of candidate programs and constructs a set of queries that, if answered by the user, would
        best resolve the ambiguities.
\end{enumerate}

We give illustrations of different kinds of queries for real-world PBE domains that can be easily answered by the user with little
cognitive load and that help reduce the ambiguity significantly.
To avoid asking too many queries, we associate queries with a \emph{disambiguation score} that represents the benefit of asking that query
for ambiguity resolution.
We generate feedback only if this score exceeds a certain threshold.
Our experimental results illustrate that our strategy for choosing the disambiguation score and the threshold leads to few false positives
and almost no false negatives.

\paragraph{Contributions}
This paper makes the following contributions:
\begin{itemize}[topsep=2pt, itemsep=2pt]
    \item We formally define the general problem of interactive program synthesis, extending upon the CEGIS~\cite{sketch,ogis},
        SyGuS~\cite{sygus}, and FlashMeta~\cite{flashmeta} formalisms.
    \item We propose an approach for incremental synthesis that leverages the observation that VSAs can be viewed as DSLs.
        We present experimental results on its performance benefits.
    \item We show how to model step-based interaction as part of the general program synthesis paradigm.
        We present experimental analysis of its effectiveness in improving synthesis convergence.
    \item We present various feedback strategies for helping users in refining their intent.
        We present qualitative and quantitative results on the
        effectiveness of these strategies in ensuring that the
        synthesized programs are correct.
\end{itemize}

\paragraph{Structure}
This paper is structured as follows.
\Cref{sec:background} provides some background on the problem of PBE,
and introduces the \ff, \fe, and \fs DSLs, which are used as running examples throughout
the paper.
\Cref{sec:overview} gives a high-level overview of our interactive synthesis formulation.
The next 3 sections formally describe each problem dimension: incremental (\Cref{sec:incremental}), step-based (\Cref{sec:step}),
and feedback-based (\Cref{sec:feedback}) synthesis.
\Cref{sec:evaluation} presents our evaluation for each dimension.
Finally, \Cref{sec:related} reviews related work, and \Cref{sec:conclusion} concludes.

%% file: background.tex
\section{Background}
\label{sec:background}

In this section, we introduce three DSLs that are used as case studies in the paper, and provide some background on inductive synthesis.

\subsection{Domain-Specific Language}
We follow the formalism of FlashMeta (a.k.a. the \projectName framework)~\cite{flashmeta}.
A synthesis problem is defined for a given \emph{domain-specific language} (DSL) $\dsl$.
A DSL is specified as a context-free grammar (CFG), with each nonterminal symbol $N$ defined through a set of \emph{rules}.
Each rule is an application of an \emph{operator} to some symbols of $\dsl$.
All symbols and operators are \emph{typed}.
If $N := F(N_1, \dots, N_k)$ is a grammar rule and $N\colon \tau$, then the output type of $F$ must be $\tau$.
A DSL has a designated \emph{output symbol} $\start{\dsl}$, which is a start nonterminal in the CFG of $\dsl$.

Every (sub-)program $P$ rooted at a symbol $N\colon \tau$ in $\dsl$ maps an \emph{input state}\footnote{DSLs in PBE typically do not
    involve mutation, so an input $\state$ is technically an \emph{environment}, not a \emph{state}.
    We keep the term ``state'' for historical reasons.} $\state$ to a value of type $\tau$.
A state is a mapping of free variables $\freeVariables{P}$ to their bound values.
Variables in a DSL are introduced by \texttt{let} definitions and $\lambda$-functions.
The output symbol has a single free variable---an \emph{input symbol} $\inputSymbol{\dsl}$ of the DSL.
For brevity, we use the notation $N [x := v]$ for ``$\mathtt{let}\ x = v\ \mathtt{in}\ N$''.

Every operator $F$ in a DSL $\dsl$ has some executable semantics.
Many operators are generic, and typically reused across different DSLs (e.g. $\mathsf{Filter}$ and $\mathsf{Map}$ list combinators).
Others are domain-specific, and defined only for a given DSL.
Operators are assumed to be deterministic and pure, modulo unobservable side effects.

\begin{figure}
    \centering
    \begin{mdframed}[innerleftmargin=-6pt, innerrightmargin=-3pt, innertopmargin=-3pt, innerbottommargin=-3pt]
        \footnotesize
        \begin{lstlisting}[language=dsl, gobble=8]
            language FlashFill;

            @output string $start$ := $e$ | std.ITE($cond$, $e$, $start$);
            string $e$ $\hspace{1pt}$:= $f$ | Concat($f$, $e$);
            string $f$ := ConstStr($w$)
                          | let string $x$ = std.Kth($vs$, $k$) in $sub$;

            string $sub$ $\hspace{46pt}$:= SubStr($x$, $pp$);
            Tuple<int, int> $pp$ $\hspace{15pt}$:= std.Pair($pos$, $pos$);
            int $pos$ $\hspace{58pt}$:= AbsPos($x$, $k$) | RegPos($x$, $rr$, $k$);
            Tuple<Regex, Regex> $rr$ := std.Pair($r$, $r$);

            bool $cond$ := let string $s$ = std.Kth($vs$, $k$) in $b$;
            @extern[std.text.match] bool $b$; $\hspace{0.8cm}$ // $\freeVariables{b} = \left\{ s\colon \text{\textbf{\texttt{string}}} \right\}$
            @input string[] $vs$;  $\qquad$  string $w$;  $\qquad$   int $k$;   $\qquad$ Regex $r$;
        \end{lstlisting}
    \end{mdframed}
    \vspace{-0.7\baselineskip}
    \caption{\ff DSL $\ffdsl$ for string transformations in spreadsheets~\cite{flashfill}.
    Each program rooted at $start$ takes an input a spreadsheet row~$vs$ and performs a chain of \texttt{if}-\texttt{elseif} matches on some
    cells of~$vs$.
    The expression in the chosen \texttt{ITE} branch returns a concatenation of constants and input substrings.}
    \label{fig:overview:flashfill}
    \vspace{-0.3\baselineskip}
\end{figure}

\begin{figure}
    \begin{mdframed}[innerleftmargin=-6pt, innerrightmargin=-3pt, innertopmargin=-3pt, innerbottommargin=-3pt]
        \begin{lstlisting}[language=dsl, gobble=8]
            language FlashExtract.Sequence;

            @output StringRegion[] $seq$ :=
                 $\hspace{1pt}$ std.Map($\lambda\, x \Rightarrow$ std.Pair($pos$, $pos$), $lines$) $\hspace{0.15cm}$        // LinesMap
                | std.Map($\lambda\, t \Rightarrow$ let string $x$ = GetSuffix($d$, $t$) in
                                 std.Pair($t$, $pos$), $posSeq$) $\hspace{0.40cm}$ // StartSeqMap
                | std.Map($\lambda\, t \Rightarrow$ let string $x$ = GetPrefix($d$, $t$) in
                                 std.Pair($pos$, $t$), $posSeq$); $\hspace{0.35cm}$ // EndSeqMap

            int[] $posSeq$ := std.FilterInt($i_0$, $k$, $rrSeq$);
            int[] $rrSeq$ $\hspace{4pt}$:= RegexMatches($d$, $rr$);
            StringRegion[] $lines$ $\hspace{14pt}$:= std.FilterInt($i_0$, $k$, $fltLines$);
            StringRegion[] $fltLines$ := std.Filter($\lambda\, s \Rightarrow b$, $allLines$);
            StringRegion[] $allLines$ $\hspace{1pt}$:= SplitLines($d$);

            @extern[std.text.match] bool $b$; $\hspace{0.72cm}$ // $\freeVariables{b} = \left\{ s\colon \text{\textbf{\texttt{string}}} \right\}$
            @extern[FlashFill] int $pos$;  $\hspace{0.90cm}$  // $\freeVariables{pos} = \left\{ x\colon \text{\textbf{\texttt{string}}} \right\}$
            @input StringRegion $d$; $\qquad$  int $i_0$;   $\qquad$ int $k$;
        \end{lstlisting}
    \end{mdframed}
    \vspace{-0.7\baselineskip}
    \caption{\fe DSL $\feseqdsl$ for selecting a \emph{sequence of spans} in a textual document $d$~\cite{flashextract}.
        Each program rooted at $seq$ is either a \textsf{LinesMap} program (split $d$ into lines and select a span in each line), a
        \textsf{StartSeqMap} program (select a sequence of starting positions of the spans and map each to its corresponding ending
        position), or a \textsf{EndSeqMap} program (select a sequence of ending positions of the spans and map each to its corresponding
        starting position).
        This DSL references position extraction logic $pos$ from $\ffdsl$ (\Cref{fig:overview:flashfill}) and string predicates $b$ from the
    standard library.}
    \label{fig:overview:flashextract}
    \vspace{-0.3\baselineskip}
\end{figure}

\begin{figure}
    \begin{mdframed}[innerleftmargin=-6pt, innerrightmargin=-3pt, innertopmargin=-3pt, innerbottommargin=-3pt]
        \begin{lstlisting}[language=dsl, gobble=8]
            language FlashExtract;

            @output TreeNode $E$ := $struct$ |  $arr$;
            ObjectNode $struct$ $\hspace{7pt}$:= Struct($E_1$, $\dots$, $E_n$) | Prop($id$, $E_r$);
            ArrayNode $arr$ $\hspace{21pt}$:= Seq($id$, $E_s$) | std.Map($\lambda\, d \Rightarrow$ $E$, $E_s$);

            StringRegion[] $E_s$ $\hspace{4pt}$:= @extern[FlashExtract.Sequence] $seq$;
            StringRegion $E_r$ $\hspace{12pt}$:= @extern[FlashFill] $sub$[$x$ := $d$];
            @input StringRegion $d$; $\qquad$ string $id$;
        \end{lstlisting}
    \end{mdframed}
    \vspace{-0.7\baselineskip}
    \caption{\fe meta-DSL $\fedsl$ for extraction of a dataset from a textual document $d$~\cite{flashextract}.
    Each program rooted at $E$ builds an object tree of \emph{sequences} and \emph{structs}, extracted from $d$.
    They are extracted using sequence selection logic from $\feseqdsl$ (\Cref{fig:overview:flashextract}) and substring selection logic from
    $\ffdsl$ (\Cref{fig:overview:flashfill}), respectively.
    The leaves of the tree are \emph{field programs}: region or sequence extractions.
    Each field program is marked with a unique ID.}
    \label{fig:overview:schema}
\end{figure}

\begin{figure}
	\begin{mdframed}[innerleftmargin=-6pt, innerrightmargin=-3pt, innertopmargin=-3pt, innerbottommargin=-3pt]
		\begin{lstlisting}[language=dsl, gobble=8]
			language FlashSplit;

			@output StringRegion[] $fields$ := SplitByDelimiters($v$, $d$);
			StringRegion[] $d$ := LookAround($v$, $c$, $rr$) | Union($d$, $d$);
			StringRegion[] $c$ := ExactMatches($v$, $s$);
                                 	| IncludeWhitespace($v$, $s$);
			Tuple<Regex, Regex> $rr$ := std.Pair($r$, $r$);

			@input StringRegion $v$; $\qquad$ string $s$; $\qquad$ Regex $r$;
		\end{lstlisting}
	\end{mdframed}
	\vspace{-0.7\baselineskip}
	\caption{\fs DSL $\fsdsl$ for \emph{splitting} an input record $v$ into a sequence of fields separated by delimiters $d$.
    Each delimiter is determined by a \textsf{LookAround} operator, which represents a constant string match $c$ in
    $v$ that also matches regular expressions $r_1$ and $r_2$ on the left and right side respectively (similar to position extraction logic
    from $\ffdsl$ in \Cref{fig:overview:flashfill}).
    Constant string matches are either exact matches of a string, or including any surrounding whitespace. }
	\label{fig:overview:flashsplit}
\end{figure}

\paragraph{\ff DSL}
\Cref{fig:overview:flashfill} shows an excerpt from the definition of \ff DSL, which transforms a list of strings (i.e., a spreadsheet row) into an output string.
In this DSL, non-prefixed operators such as \textsf{AbsPos} and \textsf{Concat} are user-defined, while namespace-prefixed operators such as \textsf{std.Kth} are defined in the standard library of \projectName.
Symbols marked as \textbf{\texttt{@extern}} reference symbols in other DSLs or in the standard library.

\paragraph{\fe DSL}
\Cref{fig:overview:schema} shows the definition of \fe DSL.
A \fe program extracts from a text document $d$ a hierarchical tree that consists of nested structs and sequences.
A leaf of this tree is either a sequence $seq$ of regions (i.e., StringRegions) or a region $sub$ within a parent node.
We refer to both $seq$ and $sub$ as a \emph{field} in the program.
All fields in $\fedsl$ are associated with some IDs. Each program corresponds to a \emph{schema} that defines the structure of the output tree.
The logic for extracting a sequence of spans within a given region is separated into a sub-DSL $\feseqdsl$ (\Cref{fig:overview:flashextract}).
The logic for extracting a subregion within a region is imported from $\ffdsl$.

\begin{figure}[t]
    \centering
    \vspace{-1.0\baselineskip}
    \includegraphics[width=\columnwidth]{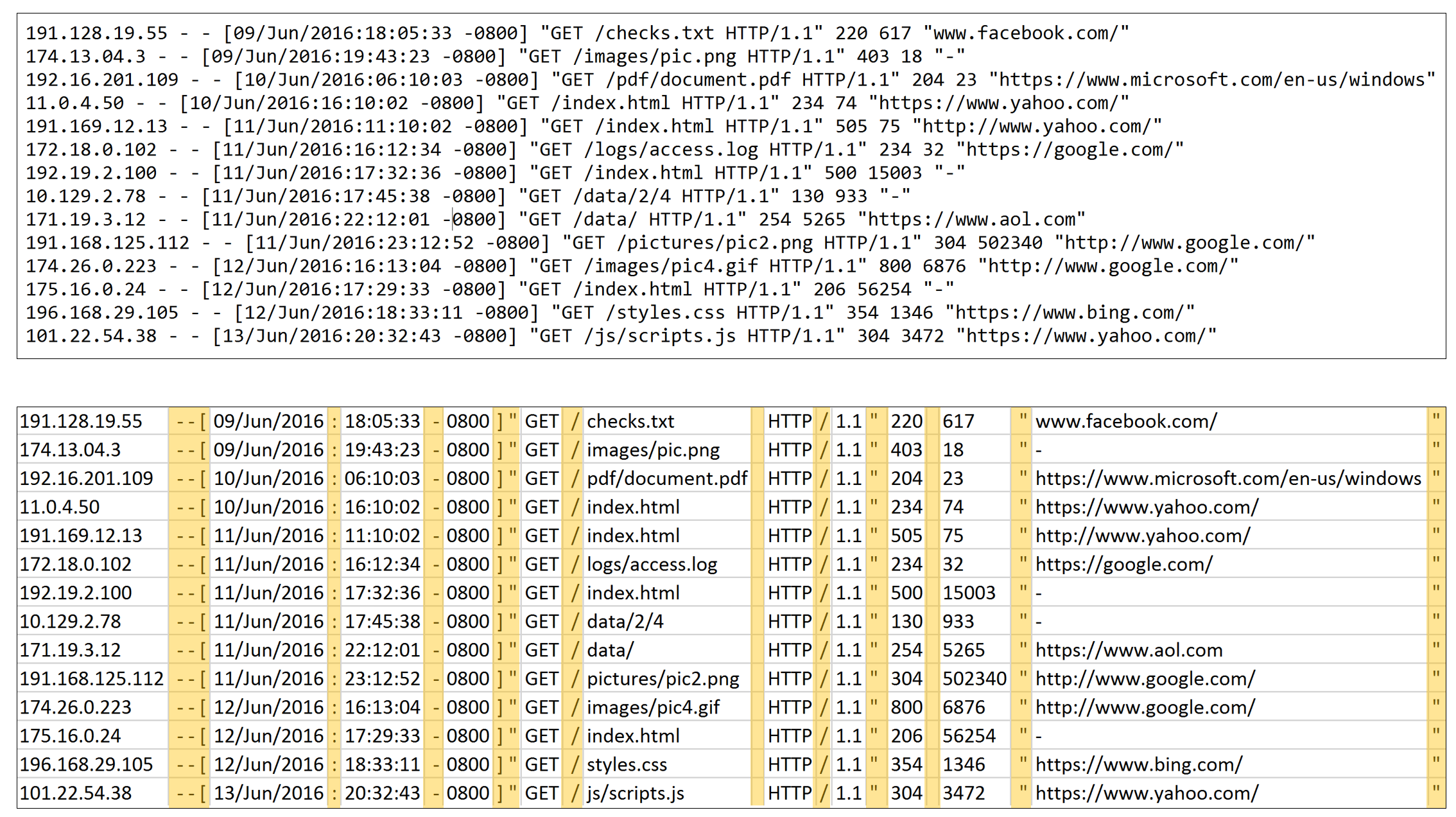}
    \vspace{-0.5\baselineskip}
    \caption{A sample splitting of a log file from a web server into fields.
        Fields are delimited by various delimiters, such as \stringliteral{- - [} and \stringliteral{] ''}.
        Note how \stringliteral{/} is used as a delimiter between some fields, but also occurs as a non-delimiter in other fields such as the URLs.}
    \vspace{0.5\baselineskip}
    \noindent\rule[0.5ex]{\linewidth}{1pt}
    \vspace{-2.5\baselineskip}
    \label{fig:split:motivation}
\end{figure}

\paragraph{\fs DSL}
    \Cref{fig:overview:flashsplit} shows the definition of the \fs DSL.
    \fs is a PBE system for field splitting in text-based data sources (e.g. log files).
    Such files often contains data rows with a large number of fields, separated by arbitrary delimiters.
    \Cref{fig:split:motivation} shows a sample splitting of a web server log.

    A \fs program splits an input data record into a sequence of field values separated by delimiters.
    In practice, a string that is used as a delimiter between some fields also occurs inside other field
    values.
    To address this, \fs supports \emph{contextual delimiters}, which match constant strings that occur
    between certain regular expression patterns on the left and right.

\subsection{Inductive Synthesis Problem}
An \emph{inductive synthesis} problem refers to synthesis of a \emph{program set} $\vsa \subset \dsl$ that is consistent with a given
\emph{inductive specification} $\spec$.
An inductive specification (or simply ``spec'') is a collection of input-output constraints $\left\{ \state_i \tospec \constraint_i
\right\}_{i=1}^n$.
Each \emph{constraint} $\constraint_i$ is a unary Boolean predicate over the \emph{output} of the desired program on the corresponding input
state $\state_i$.
The simplest form of $\constraint$ is a concrete output value; in this case, $\spec$ is a collection of \emph{input-output examples}.
In this work, we use PBE and inductive synthesis interchangeably.

A program $P$ \emph{satisfies} a spec $\spec$ (written $P \models \spec$) iff it satisfies all constraints $\state_i \tospec \constraint_i$ in $\spec$.
A program $P$ satisfies an input-output constraint $\state \tospec \constraint$ iff its output $\semantics{P}{\state}$ on the given
input state $\state$ satisfies the corresponding constraint predicate $\constraint$, i.e. if $\constraint(\semantics{P}{\state})$ is true.
A program set $\vsa$ is said to be \emph{valid} w.r.t. a spec $\spec$ (written $\vsa \models \spec$) iff all programs $P \in \vsa$ satisfy
$\spec$.

A \emph{synthesis algorithm}, given a symbol $N \in \dsl$ and a spec $\spec$, learns a \emph{valid} program set $\vsa$ of
programs rooted at $N$.
We denote the corresponding problem definition as $\mathsf{Learn}(N, \spec)$.
By definition, the algorithm must be \emph{sound}---every program in $\vsa$ must satisfy $\spec$.
A synthesis algorithm is said to be \emph{complete} if it learns \emph{all possible} programs in $\dsl$ that satisfy $\spec$.
Since $\dsl$ is usually infinite because of unrestricted constants, completeness is usually defined w.r.t. some
finitization of $\dsl$ (possibly dependent on a given synthesis problem).

\begin{example}
    A typical spec for \ff synthesis is shown below:
    \[
        \spec =
        \Biggl\{
        \begin{aligned}
            &\left\{vs \mapsto [\stringliteral{323-708-7700}]\right\} \tospec \stringliteral{(323) 708-7700} \\
            &\left\{vs \mapsto [\stringliteral{555.988.0139}]\right\} \tospec \stringliteral{(555) 988.0139}
        \end{aligned}
    \]
    It consists of 2 constraints $\constraint_1$ and $\constraint_2$.
    Each constraint $\constraint_i$ is an \emph{example constraint}: it maps a single input state $\state_i$ to the desired \ff output
    string $o_i$.
    In each input state $\state_i$, the input variable $vs$ of $\ffdsl$ is bound to a single input spreadsheet cell.
\end{example}

\begin{example}
    A typical spec for \fe sequence synthesis is shown below:
    \[
        \spec =
        \left\{ d \mapsto \stringliteral{Brazil 23 21\textbackslash n Bulgaria 35 32\textbackslash n} \right\} \tospec
        [\stringliteral{Brazil}, \dots]
    \]
    It contains a single \emph{prefix constraint} $\constraint$ with an input document $d$ in the state $\state$ and a
    \emph{prefix} of the desired sequence of selections.
\end{example}

\subsection{Version Space Algebra}
A \emph{version space algebra} (VSA) is a data structure for efficient storage of candidate programs in deductive synthesis.
Since deductive synthesis typically works with large program sets (up to $10^{50}$ programs), it requires a special data structure to
represent them in polynomial space and perform polynomial-time set operations.
We refer the reader to~\cite[\S 4]{flashmeta} for a detailed overview of VSAs; this section provides only a brief background.

\begin{defn}[Version space algebra]
    \label{def:vsa}
    Let $N$ be a symbol in a DSL $\dsl$.
    A \emph{version space algebra} is a representation for a set $\vsa$ of programs rooted at $N$. The grammar of VSAs is:
    \vspace{1pt}
    \begin{align*}
        \vsa \is \left\{P_1, \dots, P_k\right\} \palt \vsaunion(\vsa_1, \dots, \vsa_k) \palt \vsajoin(\vsa_1, \dots, \vsa_k)
    \end{align*}
    where $F$ is any $k$-ary operator in $\dsl$, and $P_j$ are some programs in $\dsl$.
    The semantics of VSA as a set of programs is given as follows:
    {\small \begin{align*}
        &\hspace{-0.5\parindent} P \in \left\{P_1, \dots P_k\right\} \tag*{if $\exists\, j\colon P = P_j$} \\
        &\hspace{-0.5\parindent} P \in \vsaunion(\vsa_1, \dots, \vsa_k) \tag*{if $\exists\, j\colon P \in \vsa_j$} \\
        &\hspace{-0.5\parindent} P \in \vsajoin(\vsa_1, \dotsc, \vsa_k) \tag*{if $P = F(P_1, \dots,
            P_k) \wedge \forall j\colon P_j \in \vsa_j$}
    \end{align*}}
\end{defn}

Intuitively, a VSA is a DAG where each node represents a set of programs.
\emph{Leaf nodes} contain explicit enumerations of programs; they are composed into larger sets by two possible VSA constructor nodes.
\emph{Union nodes} represent a set union of their constituent VSAs.
\emph{Join nodes} represent a cross-product of their constituent VSAs, with an associated operator $F$ applied to all combinations of
parameter programs from the cross-product.

VSAs support multiple efficient set operations.
In PBE we typically make use of: intersection $\vsa_1 \cap \vsa_2$, clustering based on program outputs on a given input $\clustering$,
ranking w.r.t. a scoring function $\mathsf{Top}_h(\vsa, k)$, and projection (filtering) onto a subset of programs satisfying a given spec
$\projection$.

\subsection{Backpropagation}
In the FlashMeta formalism, the main synthesis algorithm typically employed for PBE is \emph{backpropagation}, or \emph{deductive
synthesis}.
It follows the grammar of $\dsl$ \emph{top-down}, applying the principle of divide-and-conquer.
At each step, it reduces the synthesis problem $\mathsf{Learn}(N, \spec)$ to simpler subproblems: either on parameters of the symbol
$N$, or on subexpressions of the spec $\spec$.

Deductive synthesis makes use of small domain-specific procedures called \emph{witness functions}.
They \emph{backpropagate} constraints on a program $F(N_1, \dots, N_k)$ to deduced constraints on its subexpressions $N_1, \dots, N_k$.
Two kinds of witness functions exist: conditional and non-conditional.
Non-conditional witness functions take as input a spec $\spec$ on $F$ and transform it into a spec on their respective parameter $N_i$.
Conditional witness functions take as input a spec~$\spec$ on $F$ with a bound value $v_j$ of some other parameter $N_j$, and transform
$\spec$ into a spec on parameter $N_i$ under the assumption that $\semantics{N_j}\state = v_j$.
Intuitively, conditional witness functions introduce branching in the top-down search process of deductive synthesis.
They split the search space into disjoint partitions based on possible outputs of a target subprogram rooted at $N_j$, and then continue
with synthesis in each partition independently.

At a high level, deductive synthesis solves a problem $\mathsf{Learn}(N, \spec)$ via a combination of 3 problem reduction techniques
(see \cite[\S 5]{flashmeta} for a detailed presentation):

\begin{description}
    \item[Split on alternatives:]
        If $N := N_1 \palt N_2$, then the algorithm solves the subproblems $\mathsf{Learn}(N_1, \spec)$ and $\mathsf{Learn}(N_2, \spec)$,
        and takes a union of results: $\vsa = \vsa_1 \vsaunionop \vsa_2$.
    \item[Backpropagation of witnesses:]
        If $N := F(N')$, then the algorithm invokes the corresponding \emph{witness function} $\omega_{N'}$ for~$N'$ in~$F$.
        The witness function transforms the spec $\spec$ into a necessary (and often sufficient) spec $\spec'$ on $N'$.
        The algorithm then solves the subproblem $\mathsf{Learn}(N', \spec')$.

        If $\omega_{N'}$ is precise (i.e., $\spec'$ is sufficient), then any valid program $P' \in \vsa'$, when used as an argument
        $F$, produces a valid program $F(P') \models \spec$.
        Thus, the program set $\vsajoin(\vsa')$ is a valid solution for $\spec$.
        If $\omega_{N'}$ is imprecise (i.e., $\spec'$ is only necessary), then the algorithm returns a projection
        $\projection[\vsajoin(\vsa')]$.

    \item[Split on conditional execution:]
        It is often impossible to construct a precise witness function for a parameter program under all possible spec conditions.
        However, it is usually possible assuming additional restrictions on \emph{other} parameters of the same operator.
        If $N := F(N_1, N_2)$, often $F$ permits two backpropagation procedures: a simple witness function $\omega_1(\spec)$ for the first
        parameter, and a \emph{conditional} witness function $\omega_2(\spec \assuming \semantics{N_1}\state = v)$ for the second parameter.
        The function $\omega_2$ produces a spec $\spec_2$ for $N_2$ that is necessary (and often sufficient) to satisfy $\spec$
        \emph{under the assumption} that the program chosen for $N_1$ evaluates to $v$.

        In such situation, the algorithm first invokes $\omega_1$ and solves the produced subproblem $\mathsf{Learn}(N_1, \spec_1)$.
        It then \emph{clusters} $\vsa_1$ on the given inputs, and splits the search into independent branches, one per cluster (i.e. one per
        each possible output of programs from $\vsa_1$).
        Within each branch, it operates under the assumption that all programs in a cluster $\vsa_{1i} \subset \vsa_1$ produce the same
        concrete value $v_i$ as an output.
        Given that value, the function $\omega_2$ produces a spec $\spec_{2i}$ for $N_2$, and the algorithm solves the subproblem
        $\mathsf{Learn}(N_2, \spec_{2i})$.
        The final result is a union of solution sets over all branches: $\vsa = \bigvsaunion_i \vsajoin(\vsa_{1i}, \vsa_{2i})$, valid by
        construction.
\end{description}

%% file: overview.tex
\section{Overview}
\label{sec:overview}

A typical workflow of a PBE session with an end user follows a flowchart shown in \Cref{fig:overview:old-workflow}.
The synthesis system is parameterized with (i) \emph{a DSL} $\dsl$, which defines a search space for target programs, and (ii) \emph{a
ranking function} $h$, which resolves ambiguity between multiple program candidates.
The user communicates her intent to the system in the form of input-output examples (or, more generally, constraints) $\spec$ for the
desired program.
The system performs a search in $\dsl$ for the subset of programs that are consistent with $\spec$, ranks them w.r.t. $h$, and returns
top-ranked candidate program(s) to the user.
The user inspects the program(s), and, if it does not match her desired behavior, refines the spec $\spec$ by introducing additional
examples (or constraints), without any help from the synthesizer.
The system now searches for a subset of programs in $\dsl$ that are consistent with the new spec $\spec'$.
This cycle continues until either (a) the user is satisfied with the current program, or (b) the system discovers that the
current spec is unsatisfiable in $\dsl$.

\Cref{fig:overview:old-workflow} and similar flowcharts in synthesis literature implement different variants of \emph{counterexample-guided
inductive synthesis} (CEGIS)~\cite{sketch}, a common inductive synthesis technique.
While effective in many applications, we found this workflow lacking in several aspects when applied on a \emph{mass-market industrial
scale} with an end user playing the role of an \emph{oracle}.
We outline our observations and solutions to associated problems below.

\paragraph{Incrementality}
In conventional CEGIS, the learner uses the refined spec $\spec'$ at each iteration to synthesize a new valid program set $\vsa \models
\spec'$.
All the information accumulated in the synthesis session is contained in $\spec'$ as a conjunction of provided examples, counterexamples,
and more general constraints.
Typically in PBE, the learner takes it all into account by solving a fresh synthesis problem, searching in the DSL $\dsl$ for programs
that are consistent with all constraints in $\spec'$.
As the size of $\spec'$ grows with each iteration, this synthesis problem becomes more complex, thereby slowing down the search
process~\cite{bitvectors}.

Our key observation here is that \emph{the program set $\vsa_i$ learned at iteration $i$ can be transformed into a new DSL $\dsl'$, which
will become the search space for synthesis in iteration $i+1$ instead of $\dsl$}.
Notice that every refined spec $\spec'$ imposes an \emph{additional} restriction on the desired program.
Thus, an $i^{\text{th}}$ program set $\vsa_i$ must be a subset of the program set $\vsa_{i-1}$, learned at the previous iteration.

We developed an efficient procedure for transforming a VSA $\vsa_i$ into a DSL definition $\dsl_{i+1}$, which replaces $\dsl$ in the next
iteration of synthesis.
This replacement achieves two significant speedups.
First, the size of $\dsl_i$ is monotonically decreasing, and quickly becomes many orders of magnitude smaller than $\dsl$.
Second, at each iteration~$i$ we are only searching for programs consistent with the latest introduced constraint $\constraint_i$, since the
DSL $\dsl_i$ by construction only contains programs that satisfy previous constraints $\constraint_1, \dots, \constraint_{i-1}$.

\paragraph{Step-based formulation}
In conventional CEGIS, the user is limited to providing constraints on the overall behavior of the desired program.
Apart from the current candidate program $P$, the user does not have any insight into the learner and its configuration (i.e., $\dsl$ and
$h$).
Thus, the user's guidance is limited to \emph{counterexamples} to the candidate program, which the learner includes in the refined
spec~$\spec'$.
As the number of refining iterations grows, so does the user's frustration, since she cannot influence the debugging experience without any
understanding of the learner.

Many synthesis tasks have a notion of \emph{sub-tasks}, on which the user can easily and naturally specify individual constraints.
One way to model this is to support synthesis of various sub-tasks as independent synthesis tasks in respective sub-DSLs.
In that case management of the composition of those sub-tasks lies with the application layer on top of the synthesis sub-systems.
This is non-trivial, sophisticates the application logic, and is often implemented in an application-specific manner
disallowing code re-use.


In our formalism, we model such interaction by allowing the user to provide constraints on \emph{named subexpressions} in a \emph{compound
DSL} instead of just top-level constraints.
These named subexpressions correspond to aforementioned sub-tasks (or ``steps'') in the interaction process.
Analogously to bottom-up programming, we allow the user to define building blocks of the target program \emph{step-by-step} first, before
assembling them into a larger expression.

In addition to simplifying synthesis application development, step-based
formalism also improves the debugging experience during learner-user interaction.
In conventional CEGIS, when the current candidate program is incorrect, the user has to analyze the behavior of the whole program, come up
with a counterexample, and communicate it to the learner, starting a new iteration.
In contrast, with the step-based formulation she can focus on learning individual named subexpressions in the program first.
In addition to reducing the scope of required reasoning, it also allows the user to ``lock'' individual subexpressions as correct.
The synthesizer can then leverage knowledge about their behavior when learning other subexpressions of the desired program, thereby
completing the entire synthesis task in fewer iterations.
In our evaluation (\Cref{sec:step-based-eval}) we have verified that the step-based formulation significantly reduces the number of
examples for many real-world synthesis tasks.



\tikzstyle{flowchart} = [semithick, align=center, >=stealth', every path/.style={->, font=\small}]
\tikzstyle{actor} = [draw, rounded corners, inner sep=6pt, node distance=2cm, font=\sffamily, outer sep=3pt, fill=blue!20]
\tikzstyle{data} = [inner sep=5pt]
\tikzstyle{empty} = [inner sep=0pt, outer sep=0pt]
\tikzstyle{pass} = [densely dashed, thin, >=spaced open triangle 45, color=gray]
\tikzstyle{edgelabel} = [sloped, anchor=center, above]

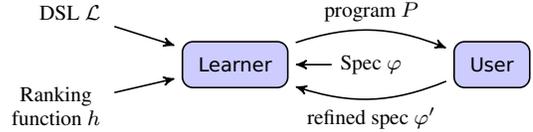
\begin{figure}
    \centering
    \begin{tikzpicture}[flowchart]
        \node[actor] (learner) {Learner};
        \node[actor, right=of learner] (user) {User};
        \node[data, below left=-0.3cm and 0.8cm of learner] (ranking) {Ranking \\ function $h$};
        \node[data, above left=0.0cm and 0.8cm of learner] (dsl) {DSL $\dsl$};
        \draw (dsl) to (learner);
        \draw (ranking) to (learner);
        \draw (learner) to["program $P$", bend left=20] (user);
        \path (user) to node[midway] (mid) {Spec $\spec$} (learner);
        \draw (mid) to (learner);
        \draw (user) to["refined spec $\spec'$", bend left=20, auto=left] (learner);
    \end{tikzpicture}
    \caption{Learner-user communication in conventional PBE.}
    \vspace{-1\baselineskip}
    \label{fig:overview:old-workflow}
\end{figure}
\begin{figure}
    \centering
    \begin{tikzpicture}[flowchart]
        \node [actor] (learner) {Learner};
        \node [actor, right=of learner] (user) {User};
        \path (learner) to coordinate[midway] (mid) (user);
        \node [actor, above=2.0cm of mid] (analyzer) {\xmakefirstuc{\hypothesizer}};
        \node [data, left=of learner] (ranking) {Ranking \\ function $h$};
        \draw (ranking) to (learner);
        \draw (learner) to ["VSA $\vsa$" {edgelabel, pos=0.33}, bend left=20]
            node[empty, pos=0.85] (lstart) {}
            (analyzer);
        \draw (user) to["spec $\spec$"]
            node[empty, pos=0.85] (rend) {}
            (learner);
        \node [data, above left=1.2cm and 1.0cm of learner] (dsl) {DSL $\dsl$};
        \draw (dsl) to
            node[empty, midway] (lend) {}
            (learner);
        \draw (analyzer.south east) to["question $q$" edgelabel, bend left=10] (user);
        \draw (user) to["response $r$" {edgelabel, below, pos=0.37}, bend left=10]
            node[empty, pos=0.88] (rstart) {}
            ($(analyzer.south east) - (0.4,0)$);
        \draw[pass] (rstart) to["refined spec $\spec'$"' edgelabel, bend right=20] (rend);
        \draw[pass] (lstart) to["refined DSL $\dsl'$"' edgelabel, bend right=30] (lend);
    \end{tikzpicture}
    \caption{Learner-user communication in interactive PBE.}
    \label{fig:overview:new-workflow}
\end{figure}
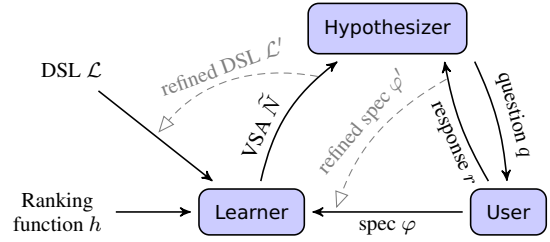

\paragraph{Feedback-based interaction}
Discovering counterexamples is relatively easy when the CEGIS oracle is a modern SMT solver (for domains that can be modeled in an SMT
theory) and when the full spec is known.
However, an end user serving as the oracle often does not have a clear understanding of the full spec, and may suffer significant cognitive
load with this task.
At every iteration, the current candidate program set $\vsa_i$ contains thousands of ambiguous programs, and humans struggle with reasoning
about possible ambiguities in intent specification.
In contrast, the synthesis system can analyze ambiguities in $\vsa_i$ and derive the most efficient way to resolve them by
proactively soliciting concrete knowledge from the user.
This observation introduces a third important actor in the CEGIS flowchart, which we call \emph{the \hypothesizer}.

Formally, any \emph{constraint type} used in PBE (e.g. example, prefix, output type) states a \emph{property} on a subset of the DSL.
Given a program set $\vsa_i$, the \hypothesizer deduces possible properties that best disambiguate among programs in $\vsa_i$.
Any such property is convertible to a Boolean or multiple-choice question $q$, which the \hypothesizer asks the user.
Any response $r$ for $q$ is convertible to a concrete constraint $\constraint$, which begins a new iteration of synthesis.

Such \emph{feedback-based} interaction has several major benefits.
First, it reduces the cognitive load on the user: instead of analyzing the program's behavior, she only answers concrete questions.
Second, it significantly reduces the number of synthesis iterations thanks to the \hypothesizer's insight into the program
set $\vsa$ and its choice of disambiguating questions.
Finally, feedback and proactiveness increases the user's confidence and trust in the system.

\subsection*{Interactive Synthesis}

\Cref{fig:overview:new-workflow} shows a typical workflow in \emph{interactive program synthesis}.
As before, the learner is parameterized with a DSL $\dsl = \dsl_0$ and a ranking function $h$.
Suppose the initial user-provided spec $\spec = \spec_1$ consists of a single constraint $\constraint_1$.

The learner synthesizes a valid program set $\vsa_1 \models \constraint_1$.
This set becomes the search space $\dsl_1$ for the next iteration of synthesis.
To refine the spec for the next iteration, the learner either waits for new constraints from the user, or proactively invokes the
\hypothesizer.
The \hypothesizer analyzes the set $\vsa_1$ and generates the best disambiguating question $q_1$ for the user.
After the user answers it with a response $r_1$, the \hypothesizer translates it to a constraint $\constraint_2$, appends it to the spec,
and invokes the next synthesis iteration with the refined spec $\spec_2$.
The learner synthesizes a valid program set $\vsa_2 \models \spec_2$ from the DSL $\dsl_1$, and the cycle
continues until convergence or unsatisfiability.

At each iteration, the user can change the context of the synthesis and specify constraints either on the overall program $P$, or on some
named subexpression $P'$ in the program.
The learner then continues synthesis for this subexpression.
From this moment, any iteration that changes a candidate program for $P'$ also triggers an incremental relearning of any subexpressions in
$P$ that are dependant on $P'$.

The challenges of conventional CEGIS described above motivate the need for modeling the interactive workflow in a first-class manner in the
program synthesis formalism.
Building on the commonly used formulations of CEGIS~\cite{sketch,ogis} and SyGuS~\cite{sygus}, we extend their problem definition to
incorporate learner-user interaction.

\begin{problem}[Interactive Program Synthesis]
    Let $\dsl$ be a DSL, and $N$ be a symbol in $\dsl$.
    Let $\synalgorithm$ be an \emph{inductive synthesis algorithm} for $\dsl$, which solves problems of type $\mathsf{Learn}(N,
    \spec)$ where $\spec$ is an \emph{inductive spec} on a program rooted at $N$.
    The specs $\spec$ are chosen from a fixed class of \emph{supported spec types} $\Phi$.
    The result of $\mathsf{Learn}(N, \spec)$ is some set $\vsa$ of programs rooted at $N$ that are consistent with $\spec$.

    Let $\spec^{*}$ be a spec on the output symbol of $\dsl$, called a \emph{task spec}.
    A~\textbf{$\bm{\spec^{*}}$-driven interactive program synthesis process} is a finite series of 4-tuples
    $\langle N_0, \spec_0, \vsa_0, \Sigma_0\rangle, \dots, \langle N_m, \spec_m, \vsa_m, \Sigma_m\rangle$, where
    \begin{itemize}[itemsep=1pt, topsep=3pt]
        \item Each $N_i$ is a nonterminal in $\dsl$,
        \item Each $\spec_i$ is a spec on $N_i$,
        \item Each $\vsa_i$ is some set of programs rooted at $N_i$ s.t. $\vsa_i \models \spec_i$,
        \item Each $\Sigma_i$ is an \textbf{interaction state}, explained below,
    \end{itemize}
    which satisfies the following axioms for any program $P \in \dsl$:
    \begin{enumerate}[itemsep=1pt, topsep=3pt, label=\textbf{\Alph*.}]
        \item $(P \models \spec^{*}) \Rightarrow (P \models \spec_i)$ for any $0 \le i \le m$;
        \item $(P \models \spec_j) \Rightarrow (P \models \spec_i)$ for any $0 \le i < j \le m$ s.t. $N_i = N_j$.
    \end{enumerate}
    We say that the process is \textbf{converging} iff the top-ranked program of the last program set in the process satisfies the task
    spec:
    \[
        P^{*} = \mathsf{Top}_h(\vsa_m, 1) \models \spec^{*}
    \]
    and the process is \textbf{failing} iff the last program set is empty: $\vsa_m = \emptyset$.

    An \textbf{interactive synthesis algorithm} $\widehat{\synalgorithm}$ is a procedure (parameterized by $\dsl$, $\synalgorithm$, and $h$)
    that solves the following problem:
    \[
        \mathsf{LearnIter}\colon \left\{
        \begin{aligned}
            \langle N_0, \spec_0, \bot\rangle &\mapsto \langle \vsa_0, \Sigma_0\rangle \\
            \langle N_i, \spec_i, \Sigma_{i-1} \rangle &\mapsto \langle \vsa_i, \Sigma_i\rangle, \quad i > 0
        \end{aligned}
        \right.
    \]
    In other words, at each iteration $i$ the algorithm receives the $i^{\text{th}}$ learning task $\langle N_i, \spec_i\rangle$ and its own
    interaction state $\Sigma_{i-1}$ from the previous iteration.
    The type and content of $\Sigma_i$ is unspecified and can be implemented by $\widehat{\synalgorithm}$ arbitrarily.
    \label{problem:interactive}
\end{problem}

\begin{defn}
    We say that an interactive synthesis algorithm $\widehat{\synalgorithm}$ is \emph{complete} iff for any task spec $\spec^{*}$:
    \begin{itemize}[itemsep=1pt, topsep=3pt]
        \item If $\exists P \in \dsl$ s.t. $P \models \spec^{*}$ then $\widehat{\synalgorithm}$ eventually converges for any
            $\spec^{*}$-driven interactive synthesis process.
        \item Otherwise, $\widehat{\synalgorithm}$ eventually fails for any $\spec^{*}$-driven interactive synthesis process.
    \end{itemize}
\end{defn}

The notion of an interactive synthesis process formally models a typical learner-user interaction where $\spec^{*}$
describes the desired program.
The general nature of definitions in \Cref{problem:interactive} allows many different implementations for $\widehat{\synalgorithm}$.
In addition to completeness, different implementations (and choices for the state $\Sigma$) strive to satisfy different \emph{performance
objectives}, such as:
\begin{itemize}[topsep=1pt, itemsep=0pt]
    \item Number of interaction rounds (e.g. examples) $m$,
    \item The total amount of information communicated by the user,
    \item Cumulative execution time of all $m+1$ learning calls.
\end{itemize}
In the rest of this paper, we present several specific instantiations of interactive synthesis algorithms that optimize these objectives.

%% file: incremental.tex
\section{Incremental Synthesis}
\label{sec:incremental}
Our incremental synthesis algorithm is based on two key ideas: \textbf{(a)} translation of VSAs as DSLs, and \textbf{(b)} local
resolution of different constraint types by memoization of intermediate subproblems.

\paragraph{VSA as a DSL}
Recall that a VSA $\vsa$ is a DAG-like program set representation with two kinds of constructor nodes (unions and joins) and one kind of
leaf nodes (explicit sets).
Our key observation here is that this representation is in fact simply an \emph{AST-based representation} of a sub-DSL $\dsl' \subset \dsl$.
More specifically, the DAG of $\vsa$ is isomorphic to a context-free grammar of a subset of $\dsl$.

\Cref{fig:incremental:algorithm} shows an algorithm for translating $\vsa$ into a grammar of $\dsl'$.
It performs an isomorphic graph translation, converting VSA unions ($\vsaunion$) into CFG alternatives ($N := N_1 \palt N_2$), VSA
joins~($\vsajoin$) into CFG operator productions ($N := F(\dots)$), and explicit program sets into CFG terminals annotated with their
possible values.

Note that as a subset of $\dsl$, the new DSL $\dsl'$ does not introduce any new operators.
Thus, all witness functions for $\dsl$ are still applicable for synthesis in $\dsl'$.
Moreover, $\dsl'$ is finite and its terminals are annotated with explicit sets of permitted values, which allows fast learning of constants
for any spec type simply via set filtering.

\begin{figure}
    \centering
    \small
    \begin{mdframed}[innerleftmargin=-8pt, innerrightmargin=3pt]
        \begin{algorithmic}[1]
            \Functionx{VsaToDsl}{VSA $\vsa$}
                \State Let $V$ be a set of fresh nonterminals, one per each non-leaf node in $\vsa$
                \State Let $\Sigma$ be a set of fresh terminals, one per each leaf node in $\vsa$
                \State \Comment{We write $\mathsf{sym}(\vsa') \in V \cup \Sigma$ to denote the corresponding fresh}
                \NoNumber \State \Comment{\,symbol for a node $\vsa'$ from $\vsa$}
                \addtocounter{ALG@line}{-1}
                \WithNumber \State Productions $R \gets \emptyset$
                \State \Comment{Create ``symbol $:=$ symbol'' productions for all union nodes}
                \ForAll{union nodes $\vsa' = \vsaunion(\vsa_1, \dots, \vsa_k)$ in $\vsa$}
                    \State $R \gets R \cup \{ \mathsf{sym}(\vsa') := \mathsf{sym}(\vsa_i) \mid i = 1 \dots k \}$
                \EndFor
                \State \Comment{Create operator productions for all join nodes}
                \ForAll{join nodes $\vsa' = \vsajoin(\vsa_1, \dots, \vsa_k)$ in $\vsa$}
                    \State $R \gets R \cup \{ \mathsf{sym}(\vsa') := F(\mathsf{sym}(\vsa_1), \dots, \mathsf{sym}(\vsa_k)) \}$
                \EndFor
                \State \Comment{Annotate terminal symbols with values extracted from leaf nodes}
                \ForAll{leaf nodes $\vsa' = \{ P_1, \dots, P_k \}$ in $\vsa$}
                    \State Annotate in $\Sigma$ that $\mathsf{sym}(\vsa') \in \{ P_1, \dots, P_k \}$
                \EndFor
                \State \Return the context-free grammar $G = \langle V, \Sigma, R, \mathsf{sym}(\vsa)\rangle$
            \EndFunction
        \end{algorithmic}
    \end{mdframed}
    \vspace{-0.5\baselineskip}
    \caption{An algorithm for translating a VSA $\vsa$ of programs in a DSL $\dsl$ into an isomorphic grammar for a sub-DSL $\dsl'
    \subset \dsl$.}
    \label{fig:incremental:algorithm}
    \vspace{-1.0\baselineskip}
\end{figure}

\paragraph{Constraint resolution}
Deductive synthesis relies on existence of witness functions, which backpropagate constraints top-down through the grammar.
Every witness function is defined for a particular \emph{constraint type}~$\constraint$ that it is able to decompose.
While some generic operators allow efficient backpropagation procedures for common constraint types (see, e.g.
\cite{feser2015synthesizing,flashextract}), most witness functions are domain-specific.

We have identified a useful set of constraint types that occur in various PBE domains and often permit efficient witness functions.
We broadly classify these constraints in 3 categories depending on their \emph{descriptive power}, with different incremental synthesis
techniques required for each category.
\begin{description}
    \item[\xmakefirstuc{\decomposable} constraints:] constructively \emph{define} a subset of the DSL by the means of backpropagation
        through witness functions.
        In other words, properties of the program's output that they describe are narrow enough to enable deductive reasoning.
        For instance:
        \begin{itemize}[itemsep=1pt]
            \item \emph{Example constraint}: ``output $= v$'',
            \item \emph{Membership constraint}: ``output $\in \left\{v_1, v_2, v_3\right\}$'',
            \item \emph{Prefix constraint}: ``output $= [v_1, v_2, \dots]$'',
            \item \emph{Subset/subsequence constraint}: ``output $\sqsupseteq [v_1, v_2, v_3]$''.
        \end{itemize}
    \item[\xmakefirstuc{\localRefine} constraints:] do not define a DSL subset on their own, but can be used to \emph{refine} an existing
        program set in a witness function for some DSL operator(s).
        For instance:
        \begin{itemize}[topsep=0pt, itemsep=1pt]
            \item \emph{Datatype constraint}: ``output$\colon \tau$''.
                Eliminates all top-level programs rooted at any type-incompatible DSL operators.
            \item \emph{Provenance constraint}: describes the desired construction method for some parts of the output.
                For example, in \ff it may take form ``substring $[i:j]$ of the output example~$o_k$ is extracted from location $\ell$ of
                the corresponding input~$\state_k$'', or ``substring $[i:j]$ of the output example $o_k$ is a date value, formatted as
                \stringliteral{YYYY-MM-DD}''.
                Allows simple domain-specific elimination of invalid subprograms.
            \item \emph{Relevance constraint}: marks inputs or parts of the input as \emph{required} or \emph{irrelevant}.
                Eliminates all programs that do not use any required parts or reference any irrelevant parts.
        \end{itemize}
    \item[\xmakefirstuc{\globalRefine} constraints:] do not define a DSL subset on their own and do not permit any efficient local
        refining logic in witness functions.
        They can only be satisfied by filtering an existing program set on the topmost level of the DSL (i.e., by projecting the set on the
        constraint).
        For instance:
        \begin{itemize}[itemsep=1pt]
            \item \emph{Negative example constraint:} ``output $\ne v$'',
            \item \emph{Negative membership constraint:} ``output $\not\ni v$''.
        \end{itemize}
\end{description}
Given a program set $\vsa$ and a new constraint $\constraint$, we filter $\vsa$ w.r.t. $\constraint$ differently depending on the category
of $\constraint$.

\begin{itemize}[itemsep=1pt, topsep=2pt]
    \item If $\constraint$ is \decomposable, it seeds a new full round of deductive synthesis, which may narrow down the set unpredictably.
        We convert the set $\vsa$ into an isomorphic DSL $\text{\textsc{VsaToDsl}}(\vsa)$, and use it as a search space for a new
        synthesis round with a spec consisting of a single constraint $\constraint$.
    \item If $\constraint$ is \localRefine, it is only relevant to select witness functions.
        Suppose these witness functions backpropagate specs for the operator $F(N_1, \dots, N_k)$.

        We first identify occurrences of $F$ in $\vsa$.
        Each node of kind $\vsajoin(\vsa_1, \dots, \vsa_k)$ in $\vsa$ has been constructed during the top-down grammar traversal in a
        previous iteration of deductive synthesis as a solution to some intermediate synthesis subproblem $\mathsf{Learn}(F(N_1, \dots,
        N_k), \spec_F)$.
        To enable incrementality, \emph{we keep references to all intermediate specs $\spec_F$ that were produced at this level in the
        previous synthesis iteration}.

        We now repeat the learning for $F$ on all retained subproblems, but with their previous specs $\spec_F$ conjoined with the new
        constraint $\constraint$.
        The witness functions for $F$ take $\constraint$ into account and produce potentially more refined specs for $N_1, \dots, N_k$.
        These new specs are \decomposable (since they were produced by witness functions), and thus initiate new rounds of incremental
        synthesis on $\text{\textsc{VsaToDsl}}(\vsa_1), \dots, \text{\textsc{VsaToDsl}}(\vsa_k)$.
        The results replace $\vsa_1, \dots, \vsa_k$ in $\vsa$ without affecting the rest of the set.
    \item If $\constraint$ is \globalRefine, it cannot be efficiently resolved by inspecting $\vsa$ or invoking witness functions.
        Thus, we have to compute the projection $\projection[\vsa][\constraint]$ at the top level.
        An efficient implementation of the projection operation computes the clustering $\clustering$ on an input $\state$ from
        the constraint $\constraint$.
        All programs in the same cluster $\vsa_k$ produce the same output $v_k$ on $\state$.
        If the constraint $\constraint$ only references the output of the desired program (as all \globalRefine constraints do), then all
        programs in $\vsa_k$ either satisfy $\constraint$ or not.
        Thus, we simply take a union of clusters where the corresponding outputs $v_k$ satisfy $\constraint$.

        We note that for many \globalRefine constraints this operation is trivial once the clustering is computed.
        For example, a negative example constraint ``output $\ne v$'' eliminates at most 1 cluster---the one where $v_k = v$, if one
        exists.
\end{itemize}

This incremental learning algorithm can be expressed as an instance of interactive synthesis formalism from \Cref{problem:interactive}.
For that, set $\Sigma_i$ to be a tuple of $\vsa_i$ (required to become the search space at the next
iteration) and all intermediate specs produced by witness functions (required to resolve \localRefine constraints).

\begin{theorem}
    If the underlying one-shot learning algorithm $\synalgorithm$ for $\dsl$ is complete, then
    \textbf{(a)} the incremental learning $\widehat{\synalgorithm}$ is complete, and
    \textbf{(b)} at each iteration $i$ the result $\vsa_i$ of incremental learning is equal to the result of
    cumulative one-shot learning $\mathsf{Learn}_{\synalgorithm}(N_i, \spec_0 \wedge \dots \wedge \spec_i)$.
    \label{th:incremental:complete}
\end{theorem}
\begin{proof}
    \vspace{-1\baselineskip}
    Omitted for lack of space.
    Intuitively, the theorem follows from the fact that $\vsa_i$ are monotonically non-increasing and from the axioms of interactive
    synthesis in \Cref{problem:interactive}.
\end{proof}

%% file: step.tex
\section{Step-based Synthesis}
\label{sec:step}

In this section we introduce our formalism for the \emph{step-based program synthesis}.
We first start with a motivational case study of \fe, then define the step-based interaction process formally, and discuss its
applicability to our DSL examples in this paper.

\subsection{Motivation}
Consider the \fe DSL $\fedsl$, presented in \Cref{fig:overview:schema}.
A program in $\fedsl$ extracts a structured dataset from a given text file.
The schema for this constructed dataset may arbitrarily combine sequences, structs, and primitive field extractions.
However, the users usually are more inclined to  provide examples for a single field and then move on to another field (as
opposed to providing examples of tuples).
They are more comfortable providing partial constraints on the individual fields of the desired output, as opposed to
burdensome complete examples of extracted objects.

This situation motivated the developers of \fe-based applications to expose constraints on individual fields as the primary UI for end
users.
They then provided these constraints as individual subproblems for sequence ($\feseqdsl$) or region ($\ffdsl$) extraction.
This required the developers to implement their own wrapper code to support the entirety of the learner-user
interaction---that is, \textbf{(a)} the schema classes for combining sequence and field programs in a compound extraction program,
and \textbf{(b)} a learning logic to guess the desired schema for a given task.
This situation had 3 drawbacks:

\begin{enumerate}
    \item Implementing (a) and (b) is cumbersome.
        Moreover, the application-specific implementations usually cannot be shared across different applications.
    \item The schema learning logic (b) is non-trivial.
        In fact, the schema format can be naturally captured using a recursive DSL in our formalism ($\fedsl$), and hence can/should be
        learned using the synthesis algorithm from the \projectName SDK.
    \item When the synthesizer treats multiple fields in the same task as independent problems, it cannot leverage insights learned from one
        field to improve its learning logic for another field.
        As a result, the entire task requires more examples to converge.
\end{enumerate}

To address these drawbacks, we define step-based synthesis over compound DSLs as a novel first-class formalism over FlashMeta.

\subsection{Problem Definition}

\begin{defn}[Compound DSL]
    A DSL $\dsl$ is called a \emph{compound DSL} if it includes any \emph{extern} nonterminals $N_1, \dots, N_m$, which resolve to output
    symbols of some \emph{sub-DSLs} $\dsl_1, \dots, \dsl_m$ respectively.
\end{defn}

\begin{defn}[Constraints on named subexpressions]
    Given a compound DSL $\dsl$, a \emph{named constraint} $\constraint^{\field\colon N}$ is specified for an extern nonterminal $N$ in
    $\dsl$.
    Its meaning is as follows:

    \begin{quote}
        \vspace{-0.5\baselineskip}
        ``There exists a subexpression rooted at $N$ in the desired compound program.
        We mark it with ID $\field$.
        This subexpression must satisfy the constraint $\constraint$.''
        \vspace{-0.5\baselineskip}
    \end{quote}
    When used in a context of iterative learner-user interaction on a larger task, all named constraints with the same ID $\field$ apply to
    the same subexpression in the desired compound program in $\dsl$.

    \emph{Named specs} $\spec^{\field\colon N}$ are defined similarly as conjunctions of named constraints on $\field$.
\end{defn}

At any point during the learner-user interaction, the step-based synthesis algorithm has accumulated the following information in its
interaction state $\Sigma_i$:
\begin{enumerate}[label=(\roman*), noitemsep, topsep=3pt]
    \item a spec $\spec$ on the compound program,
    \item a list of named subexpressions $\field_1, \dots, \field_n$, which appear in the compound program, and
    \item specs $\spec^{\field_1}, \dots, \spec^{\field_n}$ on these subexpressions.
\end{enumerate}

When the user provides a new named constraint $\constraint^{\field_i}$ on a subexpression $\field_i$, deductive synthesis must
learn a new VSA of compound programs, incorporating the new constraint in the process.
It starts the top-down synthesis from $\spec$, as before.
When the search process reaches the nonterminal $N_i$ with some deduced spec $\spec'$, the learner compares it with the new constraint
$\constraint^{\field_i}$ (since $\spec'$ is known to already be compatible with the previous named spec $\spec^{\field_i}$).
The witness functions for $N_i$ now must consider the new constraint $\constraint^{\field_i}$.

Two options exist at this point:
\textbf{(a)} we are learning the named subexpression $\field_i$, or
\textbf{(b)} we are learning an unrelated subexpression, which happens to also start from $N_i$.
Option (b) has already been considered in the previous synthesis iterations, so the learner can just reuse the corresponding VSA.
Thus, it must only detect whether option (a) is applicable.
To do that, the witness functions for $N_i$ check if the deduced spec $\spec'$ is compatible with the new
constraint~$\constraint^{\field_i}$.
If it is, they generate a refined spec $\spec''$ for subsequent learning of $\field_i$ in its sub-DSL $\dsl_i$, and the VSA for $\field_i$
is replaced with the new one.

Re-learning of $\field_i$ may impact other subexpressions of the program, which depend on $\field_i$.
In the FlashMeta formalism these dependencies are updated automatically, by the means of conditional witness functions.
If a VSA learned for a prerequisite subexpression changes, then its clustering will also change, and the new branches will produce new specs
for the dependent subexpression.

\subsection{Case Studies}

\begin{example}[\fe]
Consider the task of extracting a sequence of customer records, each of which contains a customer name and phone number from the following
text file:
\begin{alltt}\small
    Carrie Dodson
    202-555-0153
    Leonard Robledo
    945-051-0159
    \(\dots\)
\end{alltt}
\end{example}

There are three named subexpressions in the task: the customer record, the customer name, and the phone number.
Step-based synthesis allows DSL designers to build different interaction models for \fe with minimal efforts.
For instance, one can build a non-step-based model where users provide the compound spec $\spec$ (i.e., some nested records with names and phone numbers).
In this model, users have to specify the relationship of the subexpressions in the spec $\spec$.
The PowerShell cmdlet \texttt{ConvertFrom-String} adopts this model because of its command-line user interface.
In another model, users only need to provide the named specs iteratively and in a step-based manner for each of the three fields.
In particular, users learn each of these in order by increasingly giving some named field instances until the field is identified correctly.
Interactive systems such as the one by \citeauthor{flashprog}~\cite{flashprog} adopts this model.

The learner of the first model is simpler than that of the second one because the relationships among the subexpressions (and therefore the structure of the output/program) are given.
The learner expands the grammar along the path specified by the structure in $\spec$ and learns all subexpressions in $\spec$.
In contrast, at any point of time, the learner of the second model has to take into account both the current (partial) program and the new named spec to create a new (partial) program.
For instance, while learning for customer name, the learner promotes the sequence \emph{field} of customer records (which is learned in previous step) to a sequence of \emph{structs} which contains a field customer name.
Subsequently, the learning of phone number puts a new field phone number into the already constructed struct customer record.
Although the step-based, interactive model involves more work (for the DSL designer), it helps to reduce the total number of examples across all fields that the user needs to provide in an extraction task (see \Cref{sec:step-based-eval}).

\begin{example}[\ff]
Consider the task of normalizing phone numbers into the format ``(XXX) XXX-XXXX'' as follows:

\begin{center}
    \small
    \begin{tabular}{ll}
        \toprule
        \textbf{Input} & \textbf{Output} \\
        \midrule
        \texttt{485-7829} & \texttt{(133) 485-7829} \\
        \texttt{555-0175} & \texttt{(033) 555-0175} \\
        \texttt{555 0122} & \texttt{(033) 555-0122} \\
        \texttt{033 555 6694} & \texttt{(033) 555-6694}  \\
        \dots & \dots \\ \bottomrule
    \end{tabular}
\end{center}
\end{example}

In the non-step-based model, because users provide the whole compound program spec $\spec$, \ff has to keep tracks of all possible partitioning of input rows, and for each partitioning, all of its transformation programs.
Since the problem is intractable, in practice people usually limit the number of the partitions, which affects the expressiveness of the DSL.
The ranking system in this model also requires more efforts because it deals with the space of all satisfying compound programs.
For instance, from the examples it is unclear if the last row belongs to the partition that contains ``555'' (and therefore ``033'' is a constant string), or it should has its own partition (which takes ``033'' from the input).

In contrast, step-based model separates the two subproblems as two named subexpressions, and enables users to provide spec for the subproblems.
Because this model eliminates the implicit dependency between the two subproblems, it makes the problem more tractable and simplify the ranking system.

%% file: feedback.tex
\section{Feedback-based Synthesis}
\label{sec:feedback}
In this section, we formally define our proposed learner-user interaction model that leverages proactive feedback in the form of queries to
the user.
We also present its applications to specific example DSLs of this paper (\ff and \fs), and discuss practical issues involved in picking a
question, evaluating its effectiveness, and defining a stopping criterion.

\input{feedback-problem}

\subsection{Case Studies}
\label{sec:feedback:casestudies}
\input{feedback-fill}
\input{feedback-split}

%% file: feedback-problem.tex
\subsection{Problem Definition}
\label{sec:feedback:problem}

Let $\dsl$ be a DSL.
Let $\Psi$ be a set of top-level \emph{constraint types} supported by the synthesizer and witness functions for $\dsl$.
For each constraint type $\constraint \in \Psi$ we associate a \emph{descriptive question} $q$ such that a
response $r$ for this question directly constitutes an instance of~$\constraint$.
We denote such a constraint as $\Psi(r)$.
Questions $q$ can be Boolean (usually in ternary logic) or multiple-choice.
We denote the set of possible responses for $q$ as $R(q)$.

\begin{example}
    An \emph{example constraint} ``output $= v$'' corresponds to a multiple-choice question ``Is the desired output on an input $\state$
    equal to $v_1, v_2, \dots, \text{ or } v_k$?''
    A response to this question constitutes an example constraint ``output $= v_i$'' for the chosen $i$.

    A \emph{datatype constraint} ``$\text{output}\colon \tau$'' corresponds to a Boolean question ``Is the desired program a
    computation of type $\tau$?
    Yes (always), no (never), or unknown (maybe).''

    Questions, like constraints, can be domain-specific.
    In \ff, a \emph{relevance constraint} states ``an input $v_i$ must/must not appear in the program.''
    It corresponds to a Boolean question ``Should the input $v_i$ be used?
    Yes (always), no (never), or unknown (maybe).''
\end{example}

\begin{figure}[t]
    \centering
    \small
    \begin{mdframed}[innerleftmargin=-8pt, innerrightmargin=3pt, innertopmargin=2pt, innerbottommargin=2pt]
        \begin{algorithmic}[1]
            \NoNumber \Procedure{Disambiguate}{Candidate programs $\vsa$, current spec $\spec$}
                \addtocounter{ALG@line}{-1} \WithNumber
                \State Analyze the ambiguities in $\vsa$ w.r.t. $\spec$.
                \Statex Let $Q$ be a set of questions that may resolve ambiguity in $\vsa$
                \State $q^{*} \gets \argmax\limits_{q \in Q} \mathsf{ds}(q, \vsa, \spec) $
                \Comment{\parbox[t]{4.0cm}{\setstretch{0.95} Compare the disambiguation scores of all questions}}
                \If{$\mathsf{ds}(q^{*}, \vsa, \spec) < \text{threshold } T$}
                    \State \textbf{break}
                \Else
                \State Present the question $q^{*}$ to the user
                \State Let $r$ be the user's response to $q$
                \State Let $\constraint$ be the response $r$ converted into a constraint
                \State \Return $\constraint$ to the learner and invoke a new round of synthesis
                \EndIf
            \EndProcedure
        \end{algorithmic}
    \end{mdframed}
    \vspace{-0.5\baselineskip}
    \caption{The \hypothesizer's proactive disambiguation algorithm.}
    \vspace{-1.3\baselineskip}
    \label{fig:feedback:algorithm}
\end{figure}

As described in \Cref{sec:overview}, we introduce a novel component called \emph{the \hypothesizer} in the learner-user interaction model.
\Cref{fig:feedback:algorithm} shows its disambiguation algorithm.
Given a VSA $\vsa$ of current candidate programs and the current iteration's spec $\spec$, the job of the \hypothesizer is to analyze~$\vsa$
and pick the best question to resolve ambiguities in~$\vsa$.
If~$\vsa$ has no ambiguities, or if the \hypothesizer is not confident in the effectiveness of potential questions, it considers the current
candidate program $P^{*} = \mathsf{Top}_h (\vsa, 1)$ correct and does not ask any questions at this iteration.

\paragraph{Disambiguation score}
To evaluate a question's effectiveness, the \hypothesizer is parameterized with a \emph{disambiguation score} function $\mathsf{ds}(q, \vsa,
\spec)$.
Higher disambiguation scores correspond to more effective questions $q$---that is, constraints generated by answering $q$ eliminate more
incorrect programs from $\vsa$.
Since the \hypothesizer cannot predict the user's response, $\mathsf{ds}(q, \vsa, \spec)$ must represent \emph{potential
effectiveness} of $q$ for any possible outcome.

Disambiguation score functions may be domain-specific or general-purpose.
In our evaluation, we found different functions to perform well for different DSLs.
In this section, we present one efficient \emph{general-purpose} disambiguation score function, which is independent of the current
iteration's spec $\spec$ but takes into account the ranking scores of alternative candidate programs in $\vsa$.

The ranking-based disambiguation score function prefers a question that promotes higher-ranked programs:
\[
    \mathsf{ds_R}(q , \vsa, \spec) \bydef \min_{r \in R(q)} \max_{P \in \vsa_r} h(P)
\]
where $\vsa_r$ is a set projection $\projection[\vsa][\Psi(r)]$, and $h$ is a ranking function provided with the DSL.
In other words, $\mathsf{ds_R}(q , \vsa, \spec)$ is higher if every response for the question $q$ leads to a higher-ranked alternative program
among the candidates that are consistent with this response.

This disambiguation score can be efficiently evaluated for many constraint types.
For instance, calculating $\mathsf{ds_R}(q , \vsa, \spec)$ for \emph{example constraints} amounts to clustering $\vsa$ and
comparing the top-ranked programs across all clusters.
Alternatively, we can quickly compute a good approximation to $\mathsf{ds_R}(q , \vsa, \spec)$ by randomly sampling $k$ programs
from the VSA and considering only their outputs.

%% file: feedback-fill.tex
\paragraph{\ff}

Our feedback-driven synthesis for the \ff language (\Cref{fig:overview:flashfill})
uses example constraint questions. In order to efficiently generate these questions,
we sample $2000$ programs from the VSA and cluster on those, assigning the
disambiguation score $\mathsf{ds_R}$ as described above. We chose $2000$ because empirically
it was a good balance between performance and having a high probability of including
at least one program from every large cluster.

%% file: feedback-split.tex
\paragraph{\fs}


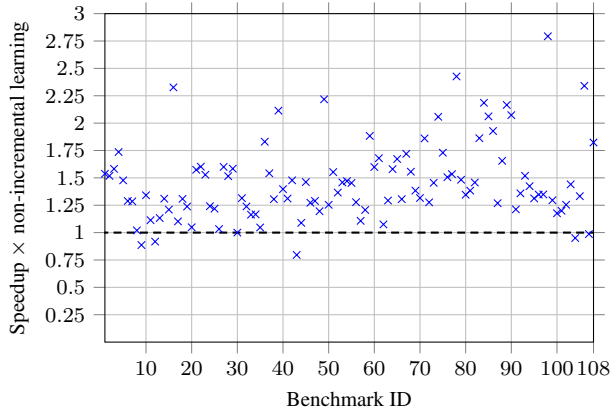
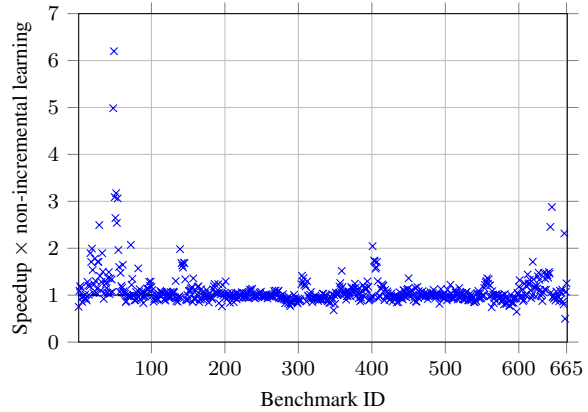
\begin{figure*}[!t]
    \centering
    \begin{subfigure}{0.48\textwidth}
        \begin{tikzpicture}
            \begin{axis}[
                    width=0.95\linewidth,
                    height=0.7\linewidth,
                    xlabel style={align=center,font=\small},
                    xlabel=Benchmark ID,
                    ylabel style={align=center,font=\small},
                    ylabel=Speedup $\times$ non-incremental learning,
                    tick align=outside,
                    grid=major,
                    xtick style={font=\small},
                    ytick style={font=\small},
                    xtick={10,20,30,40,50,60,70,80,90,100,108},
                    ytick={0.25,0.5,0.75,1.0,1.25,1.5,1.75,2.0,2.25,2.5,2.75,3.0},
                    ymin=0,
                    xmin=1,
                    ymax=3.0,
                    xmax=108,
                    legend pos = north west,
                    legend cell align=left,
                    legend style={font=\small}
                ]
                \addplot[only marks, mark=x, color=blue] table {figures/ff-incremental-speedup.dat};
                \addplot[color=black,style=densely dashed,thick] coordinates {(0,1) (108,1)};
            \end{axis}
        \end{tikzpicture}
        \caption{Speedup on the \ff DSL.}
        \label{subfigure:ff-speedup}
    \end{subfigure}
        \begin{subfigure}{0.48\textwidth}
        \begin{tikzpicture}
            \begin{axis}[
                    width=0.95\linewidth,
                    height=0.7\linewidth,
                    xlabel style={align=center,font=\small},
                    xlabel=Benchmark ID,
                    ylabel style={align=center,font=\small},
                    ylabel=Speedup $\times$ non-incremental learning,
                    tick align=outside,
                    grid=major,
                    xtick style={font=\small},
                    ytick style={font=\small},
                    xtick={100,200,300,400,500,600,665},
                    ytick={0,1,2,3,4,5,6,7},
                    ymin=0,
                    xmin=1,
                    ymax=7,
                    xmax=666,
                    legend pos = north west,
                    legend cell align=left,
                    legend style={font=\small}
                ]
                \addplot[only marks, mark=x, color=blue] table {figures/fe-incremental-speedup.dat};
                \addplot[color=black,style=densely dashed,thick] coordinates {(0,1) (108,1)};
            \end{axis}
        \end{tikzpicture}
        \caption{Speedup on the \fe DSL.}
        \label{subfigure:fe-speedup}
    \end{subfigure}
    \caption{Speedups obtained by the incremental synthesis algorithm vs. the non-incremental algorithm. Values higher above the $y=1$ line (where the runtimes are equal) are better.}
    \label{figure:incremental-learning-performance}
    \vspace{-1.5\baselineskip}
\end{figure*}

Our instantiation of the feedback-driven paradigm for the \fs language (\Cref{fig:overview:flashsplit}) intends to resolve ambiguities in
learning field-splitting programs with arbitrary delimiters.
The \fs's ranking function favors combinations of delimiters that occur regularly across all input rows and produce a uniform splitting.
The synthesis ambiguity lies in choosing a particular combination of consistently aligned delimiters constitutes the desired
program.
For example, there exists a huge number of natural ways to split the server log data in \Cref{fig:split:motivation} into fields (e.g.
separate the ``Date'' field into ``Day'', ``Month'', and ``Year'').
Examples help to resolve this ambiguity, but each data row may have up to 50 fields; thus, providing even a single complete
example of split positions in a row may be too burdensome.
We alleviate this user effort in two ways.

\subparagraph{Using constraints other than examples}
\fs supports a \emph{subset constraint} (as defined in \Cref{sec:incremental}), where the user specifies a \textrm{partial example} with
only some of the split positions on a given input.

\subparagraph{Providing feedback in the form of questions}
The \hypothesizer analyzes the program set $\vsa$ of all candidate delimiter expressions, and guides the user by asking questions about
ambiguous split positions.
We investigated two kinds of questions to elicit feedback:

\begin{description}
    \item [  Binary position questions. ]
        A binary position question $q \in Q_b$ presents a single position in the input row and asks if it is a desired splitting point.
        The answer ``Yes'' corresponds to a \emph{positive subset constraint} over desired split positions, and the answer ``No''
        corresponds to a \emph{negative membership constraint} over desired split positions.
        The \hypothesizer generates such questions when there exist positions that are unique to certain candidate delimiter expressions.
    \item [ Confirmation questions.]
        A confirmation question $q \in Q_c$ presents a set of positions to the user and asks whether all of these positions are valid
        splitting points.
        The answer ``Yes'' corresponds to a \emph{positive subset constraint} with all presented positions.
        The answer ``No'' corresponds to a $\false$ constraint, meaning that no program in $\fsdsl$ can satisfy the user's constraints.
        The \hypothesizer generates such questions when the user provides an example that can only be satisfied by one delimiter expression
        in $\vsa$.
        In this case, all of the positions determined by this delimiter expression must be part of the output, or else the system can
        declare failure early on.
        This saves the user the effort of providing all of these redundant examples individually.
\end{description}

In Section \cref{subsec:feedback:eval} we evaluate effectiveness of these question types individually and in combination.
For individual question types, since our ranking function $h$ does not distinguish among the top ranked well-aligned delimiters, we use a
uniform disambiguation score over all candidate programs for a particular question type.\footnote{$h$ can be improved by
    statistical analysis of delimiters that more commonly occur in practice.
    As we found \fs sufficiently efficient for our scenarios in evaluation, we left such analysis for future work.
}
For a combined system, we found that binary questions perform better than confirmation questions as the number of desired splits grows.
This inspires a split-based domain-specific disambiguation score:
\[
    \mathsf{ds_{FS}}(q , \vsa, \spec) \bydef
\begin{cases}
     \mathsf{MinSplits}(\vsa, \spec) - t  & \text{if } q \in Q_b\\
    t -\mathsf{ MinSplits}(\vsa, \spec)   & \text{if } q \in Q_c
\end{cases}
\]
where $\mathsf{MinSplits}(\vsa,\spec)$ is the minimum number of splits produced by any program in
$\vsa$ on any input in $\spec$, and $t$ is a predefined threshold.
We found $t=30$ to work well in our experiments.

%% file: evaluation.tex
\section{Evaluation}
\label{sec:evaluation}


\input{eval-incremental}

\input{eval-stepbased}
\input{eval-feedback}

%% file: eval-incremental.tex
\subsection{Incremental Synthesis}
We implemented the incremental synthesis algorithm, described in \Cref{sec:incremental}, in the PROSE framework.
We evaluate the incremental synthesis algorithm in the context of the \ff DSL.
For this case study, we picked all the benchmarks which required the user to provide two or more examples to learn a correct program, from
among the \ff benchmarks.
All the data reported in this subsection were obtained by repeating each experiment ten times and averaging the results after discarding outliers.

\Cref{figure:incremental-learning-performance} summarizes the results of our evaluation.
\Cref{figure:incremental-learning-performance}\subref{subfigure:ff-speedup} plots the \emph{speedup} obtained by
incremental algorithm over the non-incremental algorithm for each benchmark for the \ff DSL.
These were computed by dividing the execution time of the non-incremental algorithm by the execution time of the incremental algorithm for
each benchmark.
We observe that almost all the speedup values are greater than one, with the exceptions being extremely short-running benchmarks as mentioned
earlier.
Further, the incremental algorithm achieves a geometric mean speedup of 1.42 over the non-incremental algorithm, across the 108 benchmarks
considered for the \ff DSL.

\Cref{figure:incremental-learning-performance}\subref{subfigure:fe-speedup}
describes the result of a similar experiment with the \fe DSL, where
a total of 665 benchmarks were used for evaluation. The performance
gains in the case of the \fe benchmarks are more modest on average
than in the case of the \ff benchmarks. Our investigations revealed
that although the pruned VSAs at each iteration consist of fewer
programs, the \emph{structure} of these VSAs can be more complex. For
example, we observed VSAs where a union operation was performed on
several hundred join nodes. Pruning along each of these paths during
incremental learning sometimes results in large execution time
overheads. We note however, that (a) incremental learning provides
speedups in more than half the \fe benchmarks, sometimes over 6X, (b)
In most of the cases where incremental learning does not provide a
speedup, the execution times are within 10\% of the
non-incremental algorithm. In fact, slowdowns greater than 10\%
were observed in just 12.7\% of the benchmarks.

We conclude the evaluation of incremental learning by mentioning that
over 80\% of the \ff benchmarks required only two learning
iterations, and over 90\% of the \fe benchmarks required three
or fewer learning iterations.  Despite the relatively small number of
learning iterations, our evaluation demonstrates that
incrementality yields significant performance improvements.

%% file: eval-stepbased.tex
\begin{figure*}
    \vspace{-1.5\baselineskip}
  \centering
    \begin{tikzpicture}
    \begin{axis}[
        width=0.85\linewidth,
        height=0.25\linewidth,
        xlabel style={align=center,font=\small},
        xlabel=Benchmark ID,
        ylabel style={align=center,font=\small},
        ylabel=Average number of\\interactions per field,
        tick align=outside,
        grid=major,
        xtick style={font=\small},
        ytick style={font=\small},
        ytick={0,1,2,3,4,5,6,7,8},
        xtick={5,10,15,20,25,30,35,40,45,50,55,60,65,70,75,80,85,90,95,100},
        ymin=0,
        xmin=0,
        ymax=8,
        xmax=100,
        legend pos = north west,
        legend cell align=left,
        legend style={font=\small}
      ]
      \addplot[color=red,line width=0.866pt] table {figures/step-based-avg-interactions.dat};
      \addplot[color=blue,line width=0.866pt,style=densely dashed] table {figures/step-based-avg-interactions-baseline.dat};
      \legend{Step-based, Non Step-based}
    \end{axis}
  \end{tikzpicture}
  \caption{The average number of interactions per field across all benchmarks. Lower is better.}
    \vspace{-1.0\baselineskip}
  \label{fig:step-avg-interactions}
\end{figure*}
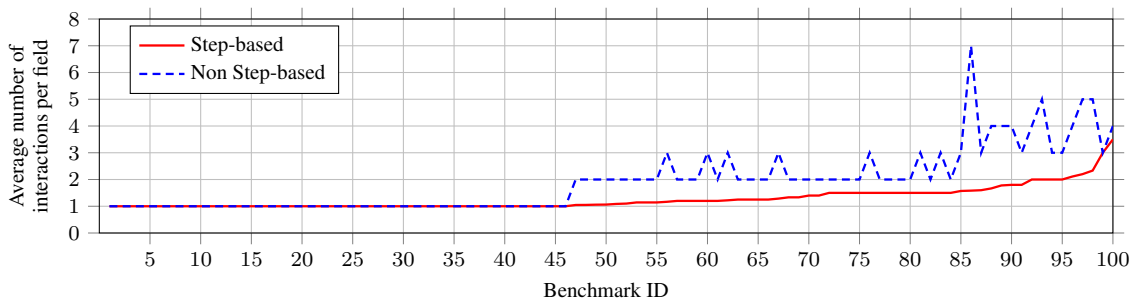

\subsection{Step-based Synthesis}
\label{sec:step-based-eval}

We use \fe to evaluate the effectiveness of step-based synthesis.
Our benchmarks consist of 100 files collected from help forums, product teams (that have exposed FlashExtract capability as a feature in their products), and their end users.
Each file corresponds to an extraction task that extracts several fields into a hierarchical output tree.
The number of extracted fields in a task ranges from 1 to 36 fields (mean 5.7, median 4).
A field contains from 1 entry (such as a customer name in an email) to a few thousand entries (such as the event timestamps in a log file).
We simulate user actions in two models: non-step-based and step-based.

\paragraph{Non-step-based \fe}
The baseline of this evaluation is a non-interactive \fe where the user has to provide examples for all the fields at once.
If extraction fails (i.e., the output is not intended), the user needs to provide new examples for each of the \emph{failing} fields.
A field fails if its output is not identical to the expected output.
The new examples of a field include all output that \fe has correctly identified in the previous interaction and the first discrepancy
between the output and the expected output of that field.
The extraction succeeds if the executing output tree is identical to the expected output tree.

\paragraph{Step-based \fe}
The user of this system extracts fields in topological order (i.e., from top-level fields to leaf fields), which is usually also the
document order.
If the current field fails, the user gives new examples until \fe produces the expected output.
The selection of new examples resembles that in the non-interactive setting, which selects the correct prefix of the output and the first
discrepancy between the output and the expected output.
Once the field extraction succeeds, the user moves to the next field.
The whole extraction succeeds if all fields are identified correctly.

\paragraph{Results}

We refer the step of locating the first output discrepancy and providing new examples for a field as an \emph{interaction}.
One of the most important goals of PBE is to reduce the number of interactions to provide better user experience.
In fact, some product teams demand that PBE systems should work with only one interaction most of the time for industrial adoption.

\Cref{fig:step-avg-interactions} shows the average number of interactions per field across all benchmarks, ordered by the number of
interactions in step-based \fe.
The evaluation shows that step-based \fe requires fewer interactions than non-step-based \fe in more than half of the benchmarks that require more than one example. 
For benchmarks that require only one interaction, step-based \fe performs similarly to non-step-based \fe because the process is entirely
non-interactive.
By dividing the extraction task into several steps, step-based \fe can ``lock'' a field if its extraction has been successful and focus on
the remaining fields.
The learning of subsequent fields therefore does not have any effects on the previously learned fields.
In contrast, non-step-based \fe has to maintain all fields as once.
When a field fails, users may have to provide examples for other fields in addition to those for the failing field.

%% file: eval-feedback.tex
\subsection{Feedback-based Synthesis}
\label{subsec:feedback:eval}

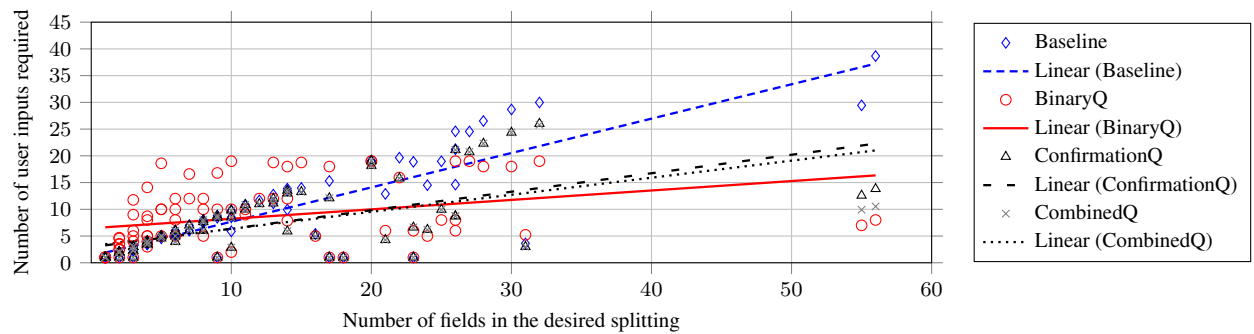
\begin{figure*}[t!]
  \centering
  \begin{tikzpicture}
    \begin{axis}[
        width=0.72\linewidth,
        height=0.27\linewidth,
        xlabel style={align=center,font=\small},
        xlabel=Number of fields in the desired splitting,
        ylabel style={align=center,font=\small},
        ylabel=Number of user inputs required,
        tick align=outside,
        grid=major,
        xtick style={font=\small},
        ytick style={font=\small},
        ytick={0,5,10,15,20,25,30,35,40,45},
        xtick={10,20,30,40,50,60},
        ymin=0,
        xmin=0,
        ymax=45,
        xmax=60,
        legend cell align=left,
        legend style={font=\small,at={(1.05,0.5)},anchor=west}
        ]
      \addplot[only marks, mark=diamond,color=blue] table {figures/split-examples-baseline.dat};
      \addplot[no marks,color=blue,line width=0.866pt,style=densely dashed] table[y={create col/linear regression={y=f(x)}}]
              {figures/split-examples-baseline.dat};
      \addplot[only marks, mark=o,color=red] table {figures/split-examples-binaryq.dat};
      \addplot[no marks,color=red,line width=0.866pt] table[y={create col/linear regression={y=f(x)}}]
               {figures/split-examples-binaryq.dat};
      \addplot[only marks, mark=triangle,color=black] table {figures/split-examples-confirmationq.dat};
      \addplot[no marks,color=black,line width=0.866pt,style=loosely dashed] table[y={create col/linear regression={y=f(x)}}]
              {figures/split-examples-confirmationq.dat};
      \addplot[only marks, mark=x,color=gray] table {figures/split-examples-combinedq.dat};
      \addplot[no marks,color=black,line width=0.866pt,style=dotted] table[y={create col/linear regression={y=f(x)}}]
              {figures/split-examples-combinedq.dat};

      \legend{Baseline, Linear (Baseline), BinaryQ, Linear (BinaryQ), ConfirmationQ, Linear (ConfirmationQ), CombinedQ, Linear (CombinedQ)}
    \end{axis}
  \end{tikzpicture}
  \caption{Comparison of effectiveness of different example provision strategies for field splitting tasks from log files. Lower is better.}
  \label{fig:feedback-split-evaluation}
    \vspace{-1.5\baselineskip}
\end{figure*}

\paragraph{\ff}
We evaluate the feedback-driven synthesis for \ff on a set of 457 text transformation tasks.
In the \emph{baseline setting}, the user provides the earliest incorrect row as the next example at each iteration.
In the \emph{feedback-driven setting}, the system instead proactively asks the user disambiguating questions on selected input rows until
the disambiguation score falls below the threshold $T$.
We set $T = 0.47$ as the mean of the score distribution over our tasks.

We evaluate \ff's feedback on two dimensions: \emph{cognitive burden} and \emph{correctness}.
Cognitive burden is defined as the number of rows the user has to \emph{read and verify} in the process.
In the baseline setting, it is the number of examples $+$ the number of correct rows before the first discrepancy that are verified
at each iteration.
In the feedback-driven setting, it is the number of questions answered.

Correctness is a combination of \emph{false positives} and \emph{false negatives}.
False positive questions occur when \ff keeps asking questions after the program is already correct.
False negative questions occur when \ff stops before it finds a correct program.

In correctness evaluation, only 2/457 tasks completed incorrectly (i.e., with false negatives).
The majority of tasks (342/457) finish with the same number of examples, and the number of false positives in the rest never exceeds 4
(specifically, 90 tasks with 1 false positive, 22 with 2, and 1 task with 4 false positives).

\Cref{tab:iterations} compares the cognitive burden of both settings.
It shows the histogram distribution of our tasks for each pair of verified row counts in baseline and feedback-driven settings.
The baseline setting often requires the user to inspect more rows (that is, the numbers \emph{below} the diagonal in \Cref{tab:iterations}
are larger than the numbers above it).

\paragraph{\fs}
We evaluate the feedback-driven synthesis for \fs on a set of 77 splitting tasks on different log files.
\Cref{fig:feedback-split-evaluation} shows the number of inputs required to complete the task against the number of split fields required by the task, for the following four example-provision strategies:

\subparagraph{Baseline}
Split position examples are provided randomly until the splitting is correct.
For each task, we average the number of examples required over 50 different random example orderings.
This models the baseline where the system does not ask any questions.
\subparagraph{BinaryQ}
One random example is provided, after which the system keeps asking binary position questions until the correct program is achieved.
This strategy is purely system-driven because the user does not provide any examples after the first.
She only answers the questions posed by the system.

\subparagraph{ConfirmationQ}
One random example is provided by the user, after which the system poses a confirmation question if one exists.
The user then provides another example, and we continue alternating between a user example and a system question until the correct splitting is achieved.
This strategy is more evenly balanced between user inputs and system feedback.

\subparagraph{CombinedQ}
The system uses a combination of binary and confirmation questions, using the disambiguation score $\mathsf{ds_{FS}}$ from \cref{sec:feedback:casestudies}.

\begin{center}
    \small
    \begin{tabular}{lc}
        \toprule
        \textbf{Strategy} & \textbf{Avg. number of inputs} \\
        \midrule
 Baseline &  8.81\\
 BinaryQ & 8.54 \\
 ConfirmationQ & 6.98 \\
 CombinedQ & 6.90\\
        \bottomrule
\end{tabular}
\end{center}

In general we see significant improvement with feedback-driven strategies over the baseline, which becomes more drastic with more fields.
Although BinaryQ performs much better than the baseline on larger splittings, it performs worse on smaller ones.
This is because for smaller splittings, the number of ambiguous programs is much larger than the number of examples requires to perform the task.
ConfirmationQ performs better on both small and large splittings.
This is because it balances user effort and system feedback: the user's examples helps reduce the search space while the system's questions
eliminate redundant examples.
Finally, with CombinedQ we see further (albeit moderate) improvement on average over the other strategies, illustrating the benefits of a system that uses different question types invoked under different circumstances.


\begin{table}[t]
    \centering
    \small
    \begin{tabular}{llllllll}
        \toprule
        & \multicolumn{7}{c}{\textbf{Feedback}} \\
        \cmidrule{2-8}
        \textbf{Baseline} & \textbf{1} & \textbf{2} & \textbf{3} & \textbf{4} & \textbf{5} & \textbf{6} & \textbf{7} \\
        \midrule
        \textbf{1} & 241 & 50 & 16 & 0 & 1 & 0 & 0 \\
        \textbf{2} & 75 & 30 & 4 & 0 & 0 & 0 & 0 \\
        \textbf{3} & 20 & 7 & 2 & 0 & 0 & 0 & 0 \\
        \textbf{4} & 0 & 5 & 3 & 0 & 0 & 0 & 0 \\
        \textbf{5} & 0 & 1 & 0 & 0 & 0 & 0 & 0 \\
        \textbf{6} & 1 & 1 & 0 & 0 & 0 & 0 & 0 \\
        \bottomrule
    \end{tabular}
    \caption{Number of rows inspected in the baseline and the feedback-driven settings for \ff evaluation.}
    \vspace{-1\baselineskip}
    \label{tab:iterations}
\end{table}

%% file: related.tex
\section{Related Work}
\label{sec:related}

\paragraph{Conversational Clarification}
In \citeyear{flashprog}, \citeauthor{flashprog} studied different user interaction models that can be applied in PBE to increase the user's
confidence in the learned program and reducing the number of iterations until convergence to the correct program~\cite{flashprog}.
They compared three interaction models: providing additional examples (positive or negative), presenting a set of candidate programs using
an English paraphrasing, and \emph{conversational clarification}, a model that prompts the user with disambiguating questions on a
discrepancy between two top-ranked candidate programs.
Among them, conversational clarification vastly outperformed the other interaction models in both convergence speed and the users'
confidence.

\citeauthor{flashprog} established that interaction (in their case, disambiguating questions) in the key to building a user-friendly
mass-market PBE system.
These results inspired us to give interactive learning first-class treatment in the PBE formalism.
In this work, we explore various important dimensions of interactive program synthesis, such as performance of synthesis iterations, impact
of different kinds of clarifying questions, and a comprehensive formalism for a step-based synthesis problem in a compound DSL.

\paragraph{Oracle-Guided Inductive Synthesis}
\citeauthor{ogis} recently developed a novel formalism for inductive synthesis called \emph{oracle-guided inductive synthesis}
(OGIS)~\cite{ogis}.
It unifies several commonly used approaches such as counterexample-guided inductive synthesis (CEGIS)~\cite{sketch} and distinguishing
inputs~\cite{jha2010oracle}.
In OGIS, an inductive learning engine is parameterized with a concept class (the set of possible programs), and it learns a
concept from the class that satisfies a given partial specification by issuing queries to a given oracle.
The oracle has access to the complete specification, and is parameterized with the types of queries it is able to answer.
Queries and responses range over a finite set of types, including membership queries, witnesses, counterexamples, verification queries, and
distinguishing inputs.
They also present a theoretical analysis for the CEGIS learner variants, establishing relations between concept classes recognizable by
learners under various constraints.

The problem of interactive program synthesis, presented in this work, can be mapped to the OGIS formalism (with the end user playing the
role of an oracle).
Hence, any theoretical results established by \citeauthor{ogis} for CEGIS automatically hold for the settings of interactive program
synthesis where we only issue counterexample queries.

In addition, inspired by our study of mass-market deployment of PBE-based systems, we present further formalism for the user's interaction
with the synthesis system.
While the ``learner'' component in the OGIS formalism is limited to a pre-defined class of queries, our formalism adds a separate modal
``\hypothesizer'' component.
Its job is to analyze the current set of candidate programs and to ask the questions that best resolve the ambiguity between programs in the
set.
The \hypothesizer is domain-specific, not learner-specific, and therefore can be refactored out of the learner and reused with different
synthesis strategies.

\paragraph{Active Learning}
In machine learning, \emph{active learning} is a subfield of semi-supervised learning where the learning algorithm is permitted to issue
queries to an oracle (usually a human)~\cite{al:settles}.
As applied to, e.g., classification problems, a typical query asks for a classification label on a new data point.
The goal is to achieve maximum classification accuracy with a minimum number of queries.
Research in active learning has focused on two important problems: (a) when to issue a query, and (b) how to pick a data point for it.

This work borrows the ideas of active learning, extends them, and applies them in the domain of program synthesis.
In our setting, the issued queries do not necessarily ask for the exact output of the desired program on a given input (an equivalent of
``label'' in ML), but may also ask for weaker output properties (e.g. verify a candidate element of the output sequence).
In all cases, though, our queries are \emph{actionable}: they are convertible to constraints, which automatically trigger a new iteration of
synthesis.
We also develop a novel approach for picking an input for the query based on its \emph{ambiguity measure} w.r.t. the current set
of candidate programs.


%% file: conclusion.tex
\section{Conclusion}
\label{sec:conclusion}

The standard user interaction model in PBE is for the user to provide constraints in an iterative manner until the user is satisfied with
the synthesized program or its behavior on the various inputs.
However, this process is far from being interactive:
\enlargethispage*{\baselineskip}

\begin{itemize}[topsep=1pt,itemsep=1pt]
    \item The constraints are over the behavior of the entire program.
        In this paper, we motivated and formalized the notion of associating constraints with sub-expressions of the program.
    \item The synthesizer is re-run from scratch with the new set of constraints.
        In this paper, we discussed how to make the synthesizer incremental, leading to a snappier UI experience for the user.
    \item The refinement of the constraints is a manual process that is guided by the user without any feedback from the synthesizer.
        In this paper, we discuss various useful feedback mechanisms.
\end{itemize}

These foundational extensions help address two key challenges of performance and correctness associated with PBE.